\documentclass[sort&compress]{elsarticle}  
\journal{XXX}
\date{\today}

\usepackage[text={15.5cm,23.5cm},centering]{geometry} 
\newcommand{\tr}{\mathrm{Tr}}

\usepackage{mathtools} 
\usepackage{amsthm}
\usepackage{amssymb}
\numberwithin{equation}{section}

\usepackage{luatex85}
\usepackage[all,poly,knot]{xy}

\usepackage[right,mathlines]{lineno}  

\usepackage{array,multirow,longtable,tabularx,booktabs}
\usepackage[flushleft]{threeparttable}
\usepackage{arydshln} 
\usepackage{lscape} %
\usepackage{pdflscape}  

\usepackage{xcolor}
\usepackage[colorlinks=true,bookmarksnumbered=true]{hyperref} 

\usepackage[shortlabels]{enumitem}  
\setlist[enumerate]{label=\upshape(\arabic*),leftmargin=*}

\usepackage{orcidlink} 
\usepackage{doi}
\usepackage[nameinlink,capitalize]{cleveref} 
\crefname{theorem}{Theorem}{Theorems}
\crefname{lemma}{Lemma}{Lemmas}
\crefname{corollary}{Corollary}{Corollaries}
\crefname{construction}{Construction}{Constructions}
\crefformat{equation}{#2\textup{(#1)}#3}
\crefmultiformat{equation}{(#2#1#3)}{ and~(#2#1#3)}{, (#2#1#3)}{ and~(#2#1#3)}

\theoremstyle{plain}
 
\newtheorem{theorem}{Theorem}[section]
\newtheorem{lemma}{Lemma}[section]
\newtheorem{corollary}{Corollary}[section]
\newtheorem{proposition}{Proposition}[section]

\theoremstyle{definition}
\newtheorem{definition}{Definition}[section]
\newtheorem{example}{Example}[section]

\theoremstyle{remark}
\newtheorem{remark}{Remark}

\begin{document}
\begin{frontmatter}
\title{Large class of many-to-one mappings over \\
quadratic extension of finite fields\tnoteref{t1}}
\tnotetext[t1]{%
{\color{purple} This paper has been published in Finite Fields and Their Applications. Please refer to this paper as: Y. Zheng, M. Zhang, Y. Ding, Z. Zha, Q. Wang, Large class of many-to-one mappings over quadratic extension of finite fields, Finite Fields Appl., vol.~118, Art.\ no.\ 102928, 2027.} \href{https://doi.org/10.1016/j.ffa.2026.102928}{https://doi.org/10.1016/j.ffa.2026.102928}
}

\author[QF,BJ]{Yanbin Zheng\textsuperscript{\orcidlink{0009-0009-1177-4515}}}
\ead{zheng@qfnu.edu.cn}

\author[QF]{Meiying Zhang}
\ead{zmy0525zmy@163.com}

\author[QF,WH]{Yanjin Ding}
\ead{596239463@qq.com}

\author[XJ]{Zhengbang Zha}
\ead{zhazhengbang@163.com}

\author[OT]{Qiang Wang\corref{cor}\textsuperscript{\orcidlink{0000-0001-5426-2776}} } 
\ead{wang@math.carleton.ca}
\cortext[cor]{Corresponding author.}

\address[QF]{School of Mathematical Sciences, Qufu Normal University, Qufu 273165, China}

\address[BJ]{Institute of Information Engineering, Chinese Academy of Sciences, Beijing 100085, China}

\address[WH]{Faculty of Mathematics and Statistics, Hubei University, Wuhan 430062, China}

\address[XJ]{College of Mathematics and System Science, Xinjiang University, Urumqi 830017, China}

\address[OT]{School of Mathematics and Statistics, 
Carleton University, Ottawa, ON K1S 5B6, Canada}

\begin{abstract}
Many-to-one mappings (including $1$-to-$1$ mappings) over finite fields play important roles in cryptography and coding theory. This paper provides a unified framework for studying the many-to-one property of polynomials of the form $f(x) = h(a x^q + b x + c) + u x^q + v x$, where $h(x) \in \mathbb{F}_{q^2}[x]$ and $a$, $b$, $c$, $u$, $v \in \mathbb{F}_{q^2}$. Our approach uses two linear transformations from $\mathbb{F}_{q^2}$ to $\mathbb{F}_{q}$ to establish a commutative diagram relating $f(x)$ to an associated polynomial $g(x)$. This diagram reduces the characterization of the many-to-one property of $f(x)$ on $\mathbb{F}_{q^2}$ to that of $g(x)$ on the subfield $\mathbb{F}_{q}$. In particular, when $h(x) = x^{r}$ and~$r$ satisfies suitable conditions, we show that $g(x)$ can be reduced to either a low-degree polynomial or a linearized binomial. Then, employing the known many-to-one properties of these low-degree or linearized polynomials over~$\mathbb{F}_{q}$, we derive explicit conditions that completely characterize when $f(x)$ is many-to-one on $\mathbb{F}_{q^2}$. Moreover, we determine the inverses of all $1$-to-$1$ mappings obtained in this paper, and we show that the $2$-to-$1$ mappings of this form naturally yield explicit involutions. Our findings generalize and unify many results in the literature. 
\end{abstract}


\end{frontmatter}

\section{Introduction}

Let~$q$ be a prime power, $\mathbb{F}_{q}$ the finite field with~$q$ elements, 
and $\mathbb{F}_{q}[x]$ the ring of polynomials over~$\mathbb{F}_{q}$. 
A polynomial $f(x) \in \mathbb{F}_{q}[x]$ is called a \textit{permutation polynomial}
(PP) of $\mathbb{F}_{q}$ if it induces a one-to-one mapping from $\mathbb{F}_{q}$ to itself.
PPs of $\mathbb{F}_{q}$ have been extensively studied; see recent surveys such as \cite{Hou15,WangIndex19,Wang25survey} and the references therein. 
In recent years, two-to-one mappings over $\mathbb{F}_{q}$ 
have been used to construct cryptographic functions 
\cite{2to1-MesQ19,KolschKK23,QuZL25JoC,ZhangLiQu26,ZhangLi2652} 
and linear codes \cite{KLiHQ214263,MesQC23719,MesQCY233285,LiLiZ26}.  
This motivates further theoretical study of two-to-one mappings and, 
more generally, of many-to-one mappings. 
To unify and generalize several earlier results 
in the literature, this paper studies the many-to-one property of 
\begin{equation}\label{eq_f=huv}
f(x) = h(a x^q + b x + c) + u x^q + v x, 
\end{equation}
where $h(x) \in \mathbb{F}_{q^2}[x]$ and $a$, $b$, $c$, $u$, $v \in \mathbb{F}_{q^2}$.
This family unifies numerous classes of  polynomials of special forms studied in the literature, including those in which $h(x)$ is a monomial, a binomial, or a sum of powers.

To derive new Kloosterman sum identities over $\mathbb{F}_{2^n}$, 
Helleseth and Zinoviev \cite{HZ03} proved in 2003 that
$(x^2 + x + c)^{r} + x$ is a PP of $\mathbb{F}_{2^n}$, 
where $\mathrm{Tr}_{2^n/2}(c) = 1$ and $r \in \{-1, -2\}$.
Inspired by their work, many researchers 
began studying this class of PPs.
From 2007 to 2013, many PPs of $\mathbb{F}_{p^n}$ of the form 
$(x^{p^k} - x + c)^{r} + L(x)$ 
were given in \cite{YD07,YDWP08,ZZH10,ZH12,LHT13},
where~$c$ and~$r$ satisfy certain conditions  
and $L(x)$ is a linearized polynomial.
In particular, PPs of $\mathbb{F}_{q^2}$ of the form 
\[
(x^{q} - x + c)^{\frac{q^2-1}{d} + 1} + x 
\]
were presented in \cite{ZH12} with $d=2$ 
and in~\cite{LHT13} with $d=3$, where $c^q = - c$. 
In 2015, Yuan and Zheng~\cite{YuanZ15} used 
the additive version of the Akbary-Ghioca-Wang (AGW) criterion \cite{AGW} 
to construct PPs of $\mathbb{F}_{q^2}$ of the form 
\[
(x^{q} + bx + c)^{\frac{q^2-1}{d} + 1} + u x^q - bx,
\]
where $b^{q+1} = 1$, $c^q = b^q c$, 
$u \in \{0, 1\}$, and $d \in \{2, 3, 4, 6\}$.
Shortly thereafter, by employing the hybrid 
additive-then-multiplicative version of the AGW criterion, 
Zheng et~al.\ \cite{ZhengFq2}  
studied PPs of $\mathbb{F}_{q^2}$ of the more general form
\begin{equation}\label{eq_fabcphi}
(ax^q + bx + c)^{r} 
\phi \big((ax^q + bx + c)^{\frac{q^2-1}{d}}\big) 
+ u x^{q} + v x, 
\end{equation}
where $a$, $b$, $c$, $u$, $v \in\mathbb{F}_{q^2}$, 
$a^{q+1} = b^{q+1}$, $a c^q = b^q c$, 
and~$d$ is an arbitrary positive divisor of $q^2-1$. 
Using the notation of \cref{eq_f=huv},  we have $h(x) = x^{r} \phi(x^{(q^2-1)/d})$ in this case.
The main contribution of \cite{ZhengFq2} is to present 
three classes of $\phi(x)$ such that 
\cref{eq_fabcphi} permutes $\mathbb{F}_{q^2}$ for any~$d$: 
$\deg(\phi) \le 2$, $\phi(x) = \sum_{k=0}^{d-1} x^k$, and $\phi(x) \in \mathbb{F}_{q}[x]$.

An underlying assumption in all the above results 
is that~$c$ satisfies certain conditions.
There are also many studies  
that impose no such condition on~$c$.
For instance, PPs of $\mathbb{F}_{q^2}$ of the form 
\begin{equation}\label{eq_1-1u1}
    (x^q - x + c)^{r} + u x^{q} + x  
\end{equation}
were investigated in 
\cite{TuZJ1531,TuZLH1534,ZhaHu16,DZheng17,
WangWuL17,WangWu18,Gupta18,ZhaHZ18,XuLC22},
where $u \in \{0, 1\}$, 
$c$ is an arbitrary element of $\mathbb{F}_{q^2}$, 
and $r = q + 2$, $2q - 1$, among other values. 
Then some PPs of $\mathbb{F}_{q^2}$ of the form
\begin{equation}\label{eq_1-1uv}
    (x^q - x + c)^{r} + u x^{q} + v x 
\end{equation}
were constructed in
\cite{XuFZ19,DZheng19,KLiQZ20,LiCao2389,ChenK+25},
where $c$, $u$, $v \in \mathbb{F}_{q^2}$ and $r = 2q$, $(q^2 - q)/2$, among other values. Clearly, \cref{eq_1-1uv} is a generalization of \cref{eq_1-1u1}.
Similarly, PPs of $\mathbb{F}_{q^n}$ of the form  
\begin{equation}\label{eq_phi+vx}
    (x^{q^k} - x + c)^{r_1} + d (x^{q^k} - x + c)^{r_2} + v x
\end{equation}
with $d \in \{0, 1\}$ were given in 
\cite{ZengZLL17,LLi18,LiuSZ18,KLiQZ20,LiCao2389,LiuJZ25Fq2}. 
Recently, some PPs of $\mathbb{F}_{q^2}$ of the form  
\begin{equation}\label{eq_1bu}
    (x^q + b x + c)^{r_1} + d (x^q + b x + c)^{r_2} - (b/e) x 
\end{equation} 
and their inverses were presented in \cite{WuY22Fq2,WuY24Fq2,WuYGL25Fq2},   
where $b$, $c \in \mathbb{F}_{q^2}$, $d \in \{0, 1\}$, 
and $e \in \mathbb{F}_{q}^{*}$.
Here $h(x) = x^{r_1} + d x^{r_2}$, $a = 1$, and $u = 0$, using the notation of \cref{eq_f=huv}. 
Moreover, some PPs of $\mathbb{F}_{q^2}$ of the form
\[
   \sum_{j = 1}^{t} \big( (x^q + x)^{r_j} + \delta_j \big)^{s_j} + u x^q + (u + v) x 
\]
were investigated in  \cite{LiWWZ20101584,LiCao24102378,ShenLX26124,QianSC26Fq2},
where $\delta_j$, $u \in \mathbb{F}_{q^2}$, 
$v \in \mathbb{F}_{q}^{*}$, and~$q$ is even.
Again, in this case, 
$h(x) = \sum_{j = 1}^{t} (x^{r_j} + \delta_j)^{s_j}$, 
$a = b = 1$, and $c = 0$.
More recently, Reis and Wang \cite{ReisWang24} 
studied a class of PPs of $\mathbb{F}_{q^n}$ of the form
$f(L(x)) + k(L(x)) M(x)$, where $f \in \mathbb{F}_{q^n}[x]$ 
is an arbitrary polynomial, 
$L, M \in \mathbb{F}_{q}[x]$ are $q$-linearized polynomials 
such that $L(x)$ divides $x^{q^n}-x$, 
and $k \in \mathbb{F}_{q^n}[x]$ satisfies a generic condition. 
Yuan et~al. \cite{Yuan25Fqn,PangYW27FqnII,PangY27Fq2,PangYWG26132}
converted the problem of determining PPs $f(x)$ of $\mathbb{F}_{q^n}$
into that of determining permutation mappings 
$F = ( f_1(x_1, \ldots, x_n), \ldots, f_n(x_1, \ldots, x_n))$
over $\mathbb{F}_{q}^{n}$. 

The main methods used to characterize these PPs 
in the above literature can be divided into two classes.
The first is to show that the equation $f(x) = \alpha$ 
has exactly one, at least one, or at most one solution 
in $\mathbb{F}_{q^2}$ for any $\alpha \in \mathbb{F}_{q^2}$, 
which is mainly used in the early papers
\cite{HZ03,YD07,YDWP08,ZZH10,ZH12,LHT13,TuZJ1531,TuZLH1534,ZhaHu16,
ZengZLL17,DZheng17,WangWuL17,WangWu18,Gupta18,ChenK+25}.
The second is to use the AGW criterion, 
which reduces the permutation behavior of a polynomial $f(x)$ over $\mathbb{F}_{q^2}$ to two conditions: the injectivity of $f(x)$ over certain subsets of $\mathbb{F}_{q^2}$,  
and the bijectivity of an associated polynomial $g(x)$ 
over a smaller subset of $\mathbb{F}_{q^2}$. 
In \cite{ZhaHZ18,LLi18,LiuSZ18,LiuJZ25Fq2,XuFZ19,XuLC22,LiCao2389,ChenK+25,LiWWZ20101584,LiCao24102378,ShenLX26124,QianSC26Fq2},
by applying the AGW criterion, the problem of whether
\cref{eq_1-1uv} permutes $\mathbb{F}_{q^2}$ 
is converted into the problem of whether an 
associated polynomial~$g(x)$ permutes the subset 
$\{x^q - x : x \in \mathbb{F}_{q^2}\}$.
In \cite{WuY22Fq2,WuY24Fq2}, 
the problem of whether \cref{eq_1bu} permutes $\mathbb{F}_{q^2}$ 
is reduced to the problem of whether an associated polynomial $g(x)$ 
permutes the subfield~$\mathbb{F}_{q}$ of $\mathbb{F}_{q^2}$.

We now review the literature on many-to-one mappings.  
In 2019, Mesnager and Qu \cite{2to1-MesQ19} provided a 
systematic study of $2$-to-$1$ mappings over finite fields. 
Later, some $2$-to-$1$ polynomials similar to \cref{eq_phi+vx} 
were used to construct involutions and binary linear codes \cite{2to1-YuanZW21,MesYZ231315,MesQCY233285}.  
From 2021 to 2023, the concept of $2$-to-$1$ 
mappings was generalized to \mbox{$n$-to-$1$} mappings in 
\cite{GAO21mto1,BartGT22194,nto1-NiuLQL23},
and some $n$-to-$1$ polynomials similar to \cref{eq_phi+vx}
were presented in \cite{nto1-NiuLQL23}. Very recently, 
Zheng et~al.\ \cite{ZhengMto1} introduced
the definition of $m$-to-$1$ mappings between two finite sets, 
which unifies and generalizes the definitions in 
\cite{2to1-MesQ19, GAO21mto1, BartGT22194, nto1-NiuLQL23},
and proposed three generic constructions of $m$-to-$1$ mappings.

We note that most of the above results focus exclusively on the permutation (i.e., $1$-to-$1$) case,  or the methods used are often tailored to specific choices of $h(x)$ or specific parameter conditions. A systematic treatment of the $m$-to-$1$ property for the broad family~\eqref{eq_f=huv} remains absent.

In this paper, we study the $m$-to-$1$ property of polynomials of the form   
\[
f(x) = h(a x^q + b x + c) + u x^q + v x  \in \mathbb{F}_{q^2}[x]
\]
using the definition of $m$-to-$1$ mappings in \cite{ZhengMto1},
where $h(x) \in \mathbb{F}_{q^2}[x]$, $a^{q+1} = b^{q+1}$, 
$a v \ne b u$, and 
$c$ is an arbitrary element of $\mathbb{F}_{q^2}$. 
First, we observe that the condition 
$a^{q+1} = b^{q+1}$ with $a b\neq 0$ 
is equivalent to  $a = d e^q$ and $b = d e$ 
for some $d$, $e \in \mathbb{F}_{q^2}^{*}$.  
We then obtain that $\tr_{q^2/q}(k f(x)) = g (\tr_{q^2/q}(e x))$,
where $g(x) = \tr_{q^2/q} \big( k h( d x + c) \big) + w x$,
$k = e^q u^q - e v^q$  and $w = u^{q+1} - v^{q+1}$.
That is, the following diagram is commutative:
\[\xymatrix{
    \mathbb{F}_{q^2} \ar[r]^{f} \ar[d]_{ \tr_{q^2/q}(e x)}   
    & \mathbb{F}_{q^2} \ar[d]^{\tr_{q^2/q}(k x)} \\
    \mathbb{F}_{q} \ar[r]^{g}       & \mathbb{F}_{q}
}\]
This diagram, 
together with one of the generic constructions 
of $m$-to-$1$ mappings from \cite{ZhengMto1}, 
yields our main theorem: 
$f(x)$ is $m$-to-$1$ on $\mathbb{F}_{q^2}$ 
if and only if $m$ divides $q$ and $g(x)$ is $m$-to-$1$ on the 
subfield~$\mathbb{F}_{q}$ of $\mathbb{F}_{q^2}$.

To demonstrate the effectiveness, we consider the case where $h(x) = x^{r}$ and~$r$ satisfies certain conditions. In particular, we specialize to the cases in which the associated polynomial $g(x)$ is a polynomial of small degree or a linearized polynomial. Using the known classification of $m$-to-$1$ polynomials of degree $\leq 4$ and $m$-to-$1$ linearized binomials over $\mathbb{F}_{q}$, we obtain a series of explicit characterizations of when $f(x)$ is $m$-to-$1$ on $\mathbb{F}_{q^2}$. 

Finally, we pay special attention to the extremal cases $m=1$ and $m=2$, which are of particular interest in cryptography and coding theory. When $f(x)$ is $1$-to-$1$, we derive a formula for $f^{-1}(x)$ in terms of $g^{-1}(x)$. For every $1$-to-$1$ mapping obtained in this paper, we can determine the inverse of the corresponding PP. Moreover, we explicitly construct involutions from the $2$-to-$1$ mappings of this form. Therefore, our results unify and generalize many earlier results.

The remainder of the paper is organized as follows.
\cref{sec_mto1} recalls the definition of $m$-to-$1$ mappings
in \cite{ZhengMto1} and reviews the $m$-to-$1$ property of 
degree $\leq 4$ polynomials and linearized polynomials.
In \cref{sec_main}, we establish a commutative diagram 
relating $f(x)$ to an associated polynomial $g(x)$, 
which reduces the problem of whether $f(x)$ is 
$m$-to-$1$ on $\mathbb{F}_{q^2}$ to the problem of
whether $g(x)$ is $m$-to-$1$ on the subfield $\mathbb{F}_{q}$. 
In Sections~\ref{sec_deg2}--\ref{sec_binomials},  
we explicitly characterize all these $m$-to-$1$ mappings 
$f(x)$ over $\mathbb{F}_{q^2}$, which are obtained from $h(x) = x^{r}$ 
such that the associated polynomial $g(x)$ is a degree $\leq 4$ polynomial or a linearized binomial. 
In \cref{sec_inv}, we determine the inverses of all these $1$-to-$1$ mappings 
and construct involutions from the $2$-to-$1$ mappings obtained in this paper.

We will use the following notation throughout this paper.
The letter $\mathbb{N}$ denotes the set of all positive integers,
$\# S$ the cardinality of a finite set~$S$,
and $\varnothing$ the empty set. 
The greatest common divisor of two integers 
$a$ and~$b$ is written as $(a, b)$. 
Let $a \bmod m$ denote the smallest non-negative 
remainder obtained when~$a$ is divided by~$m$.
The trace and norm functions from $\mathbb{F}_{q^n}$ to 
$\mathbb{F}_{q}$ are defined by
$\mathrm{Tr}_{q^n/q}(x) = \sum_{i=0}^{n-1} x^{q^i}$ and
$\mathrm{N}_{q^n/q}(x) = x^{(q^n-1)/(q-1)}$, respectively.
Since almost all trace functions involved 
in this paper are from $\mathbb{F}_{q^2}$ to $\mathbb{F}_{q}$, 
$\tr_{q^2/q}$ is simply denoted by $\tr$. 
Moreover, with a slight abuse of notation, the symbol $x$ denotes both a formal indeterminate in a polynomial ring and an element of $\mathbb{F}_{q}$; the intended meaning will always be clear from the context.

\section{\texorpdfstring{$m$-to-$1$}{m-to-1} mappings on finite fields}\label{sec_mto1}

We first recall the definition of $m$-to-$1$ 
mappings introduced in \cite{ZhengMto1} 
and slightly modified in \cite{ZhengMto1RFDeg3}, 
and then review the $m$-to-$1$ property of 
polynomials of degree $\leq 4$ and linearized polynomials. 

\subsection{Definition and construction of \texorpdfstring{$m$-to-$1$}{m-to-1} mappings}

\begin{definition}[\cite{ZhengMto1,ZhengMto1RFDeg3}]\label{defn_mto1}
  Let~$A$ be a finite set and let~$m$ be a positive integer. 
  Write $\# A = k m + r$, where~$k$ and~$r$ are integers 
  with $k \geq 0$ and $0 \leq r < m$. 
  Let~$f$ be a mapping from~$A$ to another finite set~$B$. 
  Then~$f$ is called $m$-to-$1$ from~$A$ to~$B$ if 
  there are~$k \geq 1$ distinct elements in~$B$ such that 
  each element has exactly~$m$ preimages in~$A$ under~$f$.
  The remaining~$r$ elements in~$A$ are called the
  \textit{exceptional elements} of~$f$,
  and the set of these~$r$ exceptional elements 
  is called the \textit{exceptional set} 
  of~$f$ and denoted by $E_{f}(A)$. 
  In particular, $E_{f}(A) = \varnothing$ 
  if and only if $m \mid \# A$.
\end{definition}

A polynomial $f(x) \in \mathbb{F}_{q}[x]$ is called $m$-to-$1$ on $\mathbb{F}_{q}$ if the associated mapping 
$c \mapsto f(c)$ from $\mathbb{F}_{q}$ to itself is $m$-to-$1$ on $\mathbb{F}_{q}$.
In particular, $f(x)$ is $1$-to-$1$ on $\mathbb{F}_{q}$ if and only if it permutes~$\mathbb{F}_{q}$.


\begin{example}
Let $f(x) = x^3 + x + 1  \in \mathbb{F}_{5}[x]$. 
Then~$f$ maps $0$, $1$, $2$, $3$, $4$ 
to $1$, $3$, $1$, $1$, $4$ in~$\mathbb{F}_{5}$, respectively. 
Thus~$f$ is $3$-to-$1$ on~$\mathbb{F}_{5}$ and the 
exceptional set $E_{f}(\mathbb{F}_{5}) = \{1, 4\}$.
\end{example}

\begin{example}\label{xnFq}
The monomial $x^n$ with $n \in \mathbb{N}$ is  
$(n, q-1)$-to-$1$ on $\mathbb{F}_{q}^{*}$
and 
$E_{x^n}(\mathbb{F}_{q}^{*})= \varnothing$.
\end{example}

\begin{lemma}[{\cite{ZhengMto1}}]\label{constr2}
Let $A$, $\bar{A}$, $S$, $\bar{S}$ be finite 
sets and $f \colon A \rightarrow \bar{A}$, 
$\bar{f} \colon  S \rightarrow \bar{S}$, 
$\lambda \colon A \rightarrow S$, 
$\bar{\lambda} \colon \bar{A} \rightarrow \bar{S}$ 
be mappings such that 
$\bar{\lambda} \circ f = \bar{f} \circ \lambda$, 
i.e., the following diagram is commutative:
\[
\xymatrix{
  A \ar[rr]^{f}\ar[d]_{\lambda}  &   &  
  \bar{A}  \ar[d]^{\bar{\lambda}} \\
  S \ar[rr]^{\bar{f}}            &   &  \bar{S}.}
\]
Suppose~$\lambda$ is surjective,
$\# \lambda^{-1}(s) = m_1 \, 
\# \bar{\lambda}^{-1}(\bar{f}(s))$,  
and~$f$ restricted to $\lambda^{-1}(s)$ is $m_1$-to-$1$ 
for any $s \in S$ and a fixed $m_1 \in \mathbb{N}$, where
\begin{equation*}\label{eq_lambda-1s}
  \lambda^{-1}(s) = \{a \in A : \lambda(a) = s\}    
  \quad\text{and}\quad
  \bar{\lambda}^{-1}(\bar{f}(s)) =
  \{b \in \bar{A} : \bar{\lambda}(b) = \bar{f}(s)\}.
\end{equation*}
Then for $1 \le m \le m_1 \, \#S$, $f$ is $m$-to-$1$ from~$A$ to~$\bar{A}$ if and only if 
$m_1 \mid m$, $\bar{f}$ is $(m/m_1)$-to-$1$ from~$S$ to~$\bar{S}$, and
\begin{equation}\label{eq_SumPreim(s)=r}
    \sum_{s \in E_{\bar f}(S)} 
    \# \lambda^{-1}(s) = \# A \bmod{m},
\end{equation}
where $E_{\bar f}(S)$ is the exceptional set of $\bar{f}$ 
being $(m/m_1)$-to-$1$ from~$S$ to~$\bar{S}$.
\end{lemma}

This result converts the problem 
whether~$f$ is $m$-to-$1$ into another problem 
whether~$\bar{f}$ is $(m/m_1)$-to-$1$.
It is a generalization of the AGW criterion in \cite{AGW}.

\subsection{Polynomials of degree~2 or~4}

\begin{lemma}[{\cite[Corollary~2.8]{AFF}}]\label{deg2pps}
  Let $f(x) = a x^{2} + b x + c \in \mathbb{F}_{q}[x]$.
  Then~$f(x)$ is $1$-to-$1$ on $\mathbb{F}_{q}$ if and only if 
  \textup{(1)} $a = 0$ and $b \ne 0$, or 
  \textup{(2)} $a \ne 0$, $b = 0$, and $q$ is even.
\end{lemma}

\begin{lemma}[{\cite{2to1-MesQ19}}]\label{deg2-2to1}
  Let $f(x) = a x^{2} + b x + c \in \mathbb{F}_{q}[x]$.
  Then~$f(x)$ is $2$-to-$1$ on~$\mathbb{F}_{q}$ if and only if 
  \textup{(1)} $q = 2$ and $a = b = 0$,
  \textup{(2)} $q$ is even and $a b \ne 0$, or 
  \textup{(3)} $q$ is odd and $a \ne 0$.
\end{lemma}

Combining \cref{deg2pps,deg2-2to1}  
yields the following result.

\begin{theorem}\label{deg2mto1}
  Let $f(x) = a x^{2} + b x + c \in \mathbb{F}_{q}[x]$ and $1 \le m \le q$. 
  Then $f(x)$ is $m$-to-$1$ on $\mathbb{F}_{q}$ if and only if one of the following holds:
\begin{enumerate}[\upshape(1)]
  \item $m = 1$, $a = 0$, and $b \ne 0$; 
  \item $m = 1$, $a \ne 0$, $b = 0$, and $q$ is even;
  \item $m = 2$, $a \ne 0$, $b \ne 0$, and $q$ is even;
  \item $m = 2$, $a \ne 0$, and $q$ is odd;
  \item $m = q$, $a = 0$, and $b = 0$.
\end{enumerate}
\end{theorem}

All $2$-to-$1$ polynomials of degree~$4$ over any $\mathbb{F}_{q}$
were determined by Mesnager and Qu~\cite{2to1-MesQ19}.
The result~$q$ is even will be used in this paper.

\begin{theorem}[{\cite[Theorem~35]{2to1-MesQ19}}]\label{deg4-2to1-p=2}
  Let $q = 2^n$ and $f(x) = x^4 + a_3 x^3 + a_2 x^2 + a_1 x \in \mathbb{F}_{q}[x]$. 
  Then $f(x)$ is $2$-to-$1$ on $\mathbb{F}_{q}$ if and only if one of the following holds:
  \begin{enumerate}[\upshape(1)]
    \item $a_3 = a_1 = 0$ and $a_2 \ne 0$;
    \item $a_3 = 0$, $a_1 \ne 0$, and $\tr_{2^n/2}(a_2^3/a_1^2) \ne \tr_{2^n/2}(1)$;
    \item $a_3 \ne 0$, $a_2^2 = a_1 a_3$, and $n$ is odd.
  \end{enumerate}
\end{theorem}

 \subsection{Polynomials of degree~3}

\begin{theorem}[{\cite{AFF}}]\label{deg3-1to1}
  Let $f(x) = a x^{3} + b x^{2} + c x + d \in \mathbb{F}_{q}[x]$ with $a \neq 0$.
  Then~$f(x)$ is $1$-to-$1$ on $\mathbb{F}_{q}$ if and only if
  \textup{(1)} $q \equiv 0 \pmod{3}$, $b = 0$, and $(-ac)^{\frac{q-1}{2}} \neq 1$, or 
  \textup{(2)} $q \equiv 2 \pmod{3}$ and $b^{2} = 3ac$.
\end{theorem}

\begin{theorem}[{\cite{2to1-MesQ19}}]\label{deg3-2to1}
  Any polynomial $f(x) \in \mathbb{F}_{q}[x]$ of degree~$3$ is not $2$-to-$1$ on $\mathbb{F}_{q}$ if $q \geq 7$.
\end{theorem}
 
\begin{lemma}[{\cite[Theorem~3.1]{nto1-NiuLQL23}}]\label{deg3-3to1-3n}
  Let $f(x) = x^{3} + b x^{2} + c x \in \mathbb{F}_{3^n}[x]$.  
  Then $f(x)$ is $3$-to-$1$ on $\mathbb{F}_{3^n}$
  if and only if $b = 0$ and $-c$ is a square in $\mathbb{F}_{3^n}^*$.
\end{lemma}

To deal with the case $q \neq 3^n$, we need the following lemmas.

\begin{lemma}[{\cite[Theorem~7.20]{FFGR}}] \label{lem_ax2+bx+c}
Let $f(x) = ax^2 + bx + c \in \mathbb{F}_{2^n}[x]$ with $a \neq 0$. 
Then $f(x)$ has two distinct roots in $\mathbb{F}_{2^n}$ if and only if
$b \neq 0$ and $\mathrm{Tr}_{2^n/2}(ac/b^2) = 0$.
\end{lemma}

\begin{lemma}[{\cite[Theorem~5.48]{FF}}]\label{lem_charsumx2}
  Let $f(x) = a x^2 + b x + c \in \mathbb{F}_{q}[x]$
  with $q$ odd and $a \ne 0$. 
  Let $\eta$ be the quadratic character of $\mathbb{F}_{q}$ 
  with the standard convention $\eta(0) = 0$. Then
  \[
  \sum_{e \in \mathbb{F}_{q}} \eta(f(e)) =
  \begin{cases}
    -\eta(a) & \mbox{if } b^2 \neq 4 a c, \\
    (q-1)\eta(a) & \mbox{if } b^2 = 4 a c.
  \end{cases}
  \]
\end{lemma}

\begin{theorem}\label{deg3-3to1}
Let $f(x) = a x^{3} + b x^{2} + c x + d \in \mathbb{F}_{q}[x]$ with $a \neq 0$.
Then $f(x)$ is $3$-to-$1$ on~$\mathbb{F}_{q}$ if and only if one of the following holds:
\begin{enumerate}[\upshape(1)]
    \item $q \equiv 0 \pmod{3}$, $b = 0$, and $(-ac)^{\frac{q-1}{2}} = 1$;
    \item $q \equiv 1 \pmod{3}$ and $b^2 = 3ac$;
    \item $q = 5$ and $b^2 + 2 a c =  \pm 2 a^2$.
\end{enumerate}
\end{theorem}
\begin{proof}
Item (1) follows from \cref{deg3-3to1-3n}.
If $3 \nmid q$, then 
\[
f \Big(x-\frac{b}{3a}\Big) 
= a \Big(x^3 + \frac{3ac-b^2}{3a^2} x\Big) 
    + f\Big(\frac{-b}{3a}\Big).
\]
Hence $f(x)$ is $3$-to-$1$ on~$\mathbb{F}_{q}$ if and only if $g(x) \coloneqq x^3 + \alpha x$ 
is $3$-to-$1$ on~$\mathbb{F}_{q}$, where $\alpha = (3ac-b^2)/3a^2$. 
When $q=2$, the result is clearly true.
When $q=5$, it is easy to verify that $g(x)$ is $3$-to-$1$ on~$\mathbb{F}_{q}$ 
if and only if $\alpha = \pm 1$, i.e., $b^2 + 2 a c =  \pm 2 a^2$.  

It suffices to show that when (i) $q = 2^n$ and $n \geq 2$, 
or (ii) $3 \nmid q$ and $q \geq 7$ is odd, 
$g(x)$ is $3$-to-$1$ on~$\mathbb{F}_{q}$ if and only if (2) holds. 
The sufficiency is an immediate consequence of \cref{xnFq}. 
To prove the necessity, suppose $g(x)$ is $3$-to-$1$ on~$\mathbb{F}_{q}$ 
and the exceptional set of $g(x)$ on $\mathbb{F}_{q}$ is $E_g(\mathbb{F}_{q})$.
Since $3 \nmid q$, we have $\# E_g(\mathbb{F}_{q}) = 1$ or $2$.
Then for each $x \in \mathbb{F}_{q} \setminus E_g(\mathbb{F}_{q})$, 
the equation $g(x+y) = g(x)$, or equivalently $y (y^2 + 3x y + 3x^2 + \alpha) = 0$, 
has exactly three distinct solutions in the variable~$y$. 
Clearly, $y = 0$ is a solution. 
Thus for each $x \in \mathbb{F}_{q} \setminus E_g(\mathbb{F}_{q})$,
\begin{equation}\label{eq_y2+xy}
  y^2 + 3x y + 3x^2 + \alpha  = 0   
\end{equation}
has exactly two distinct nonzero solutions in the variable~$y$. 

When $q = 2^n$ and $n \geq 2$, \cref{eq_y2+xy} 
is reduced to $y^2 + \alpha = 0$ if $x = 0$, 
which has only one solution in $\mathbb{F}_{2^n}$.
Thus $x = 0 \in E_g(\mathbb{F}_{2^n})$. 
By \cref{lem_ax2+bx+c},  
for each $x \in \mathbb{F}_{2^n} \setminus E_g(\mathbb{F}_{2^n})$, 
\cref{eq_y2+xy} has exactly two distinct nonzero solutions in $\mathbb{F}_{2^n}$ if and only if 
\begin{equation}\label{eq_tr=0anex2}
  \mathrm{Tr}_{2^n/2}(1 + \alpha/x^2) = 0  
  \qquad\text{and}\qquad 
  x^2 + \alpha \ne 0.
\end{equation}
If $n$ is even, then $2^n \equiv 1 \pmod{3}$, 
and so $E_g(\mathbb{F}_{2^n}) = \{ 0 \}$.  
Since $x^2$ permutes $\mathbb{F}_{2^n}^*$, for each 
$x \in \mathbb{F}_{2^n} \setminus E_g(\mathbb{F}_{2^n}) = \mathbb{F}_{2^n}^*$,
\cref{eq_tr=0anex2} holds if and only if $\alpha = 0$. 
If $n \geq 3$ is odd, then $2^n \equiv 2 \pmod{3}$, 
so $\# E_g(\mathbb{F}_{2^n}) = 2$. 
Suppose $E_g(\mathbb{F}_{2^n}) = \{0, e\}$, 
where $e \in \mathbb{F}_{q}^{*}$. Then
\[
\{ x^2 : x \in \mathbb{F}_{2^n} \setminus E_g(\mathbb{F}_{2^n})\} 
= \mathbb{F}_{2^n} \setminus \{0, e^2\}.
\]
For each $x \in \mathbb{F}_{2^n} \setminus E_g(\mathbb{F}_{2^n})$,
\cref{eq_tr=0anex2} holds if and only if
(i) $\alpha = 0$ and $\mathrm{Tr}_{2^n/2}(1) = 0$; 
or (ii) $\alpha = e^2$ and $\mathrm{Tr}_{2^n/2}(1 + e^2/x^2) = 0$. 
Since~$n$ is odd, we have $\mathrm{Tr}_{2^n/2}(1) = n = 1$ 
and thus (i) does not hold. 
Note that $\{1 + e^2/x^2 : x \in \mathbb{F}_{2^n} 
\setminus E_g(\mathbb{F}_{2^n})\} = \mathbb{F}_{2^n} \setminus \{0, 1\}$.
Because half of the elements in $\mathbb{F}_{2^n} \setminus \{0, 1\}$ 
satisfy $\mathrm{Tr}_{2^n/2}(x_0) = 1$, there are half of 
$x \in \mathbb{F}_{2^n} \setminus E_g(\mathbb{F}_{2^n})$ such that 
$\mathrm{Tr}_{2^n/2}(1 + e^2/x^2) = 1$. Thus (ii) does not hold.
Hence \cref{eq_tr=0anex2} holds if and only if~$n$ 
is even and $\alpha = 0$, that is, (2) holds. 

When $3 \nmid q$ and $q \geq 7$ is odd,
multiplying the both sides of \cref{eq_y2+xy} by~$4$ yields 
$(2y + 3x)^2 = -3x^2 - 4 \alpha$.
Then for each $x \in \mathbb{F}_{q} \setminus E_g(\mathbb{F}_{q})$, 
$-3x^2 - 4 \alpha$ is a nonzero square, since \cref{eq_y2+xy}
has two distinct solutions in the variable~$y$. Hence 
\[
\sum_{x \in \mathbb{F}_{q}} \eta(-3x^2-4 \alpha) \geq
\begin{cases}
q-2 & \text{if~} q \equiv 1 \pmod{3}, \\
q-4 & \text{if~} q \equiv 2 \pmod{3},
\end{cases}
\]
and so $\sum_{x \in \mathbb{F}_{q}} \eta(-3x^2-4 \alpha) \geq 3$.
Let $d = -4 (-3) (-4\alpha) = -3 \cdot 4^2 \alpha$. 
Since $3 \nmid q$ and~$q$ is odd, 
we get $d = 0$ if and only if $\alpha = 0$. 
By \cref{lem_charsumx2},
\[
\sum_{x \in \mathbb{F}_{q}} \eta(-3x^2-4 \alpha) =
\begin{cases}
-\eta(-3) & \text{if}~ \alpha \neq 0, \\
(q-1)\eta(-3) & \text{if}~ \alpha = 0.
\end{cases}
\]
If $\alpha \neq 0$, then 
$\sum_{x \in \mathbb{F}_{q}} \eta(-3x^2-4 \alpha) = - \eta(-3) \le 1$, 
a contradiction. Thus $\alpha = 0$. Then $g(x) = x^3$. 
If $q \equiv 2 \pmod{3}$, then $(3, q-1) = 1$,
and so $g(x)$ is $1$-to-$1$ on $\mathbb{F}_{q}$, 
which contradicts that $g(x)$ is $3$-to-$1$ on~$\mathbb{F}_{q}$.
Hence $q \equiv 1 \pmod{3}$. Then (2) holds. 
\end{proof}

\begin{remark}
The definition of $m$-to-$1$ mappings here extends the $n$-to-$1$ mappings in \cite{nto1-NiuLQL23}.
Hence \cite[Theorem~3.4]{nto1-NiuLQL23} 
is a special case of \cref{deg3-3to1}.
We note that the proof of \cref{deg3-3to1} is 
similar to that of \cite[Theorem~3.4]{nto1-NiuLQL23}.  
However, we use an explicit formula 
for the involved character sum rather than an estimate. 
\end{remark}

\subsection{Linearized polynomials}

We now give an explicit criterion for arbitrary $m$-to-$1$ linearized binomials.

\begin{theorem}\label{mto1_Lbi}
Let $L(x) = b x^{q^s} - c x^{q^t}\in \mathbb{F}_{q^n}[x]$ and 
$u = (s-t,\, n)$, where $b \neq 0$ and $0 \leq t \leq s$. 
If $\mathrm{Tr}_{q^n/q^u}(c/b) \neq 1$, then $L(x)$ is $1$-to-$1$ on $\mathbb{F}_{q^n}$. 
If $\mathrm{Tr}_{q^n/q^u}(c/b)    = 1$, then $L(x)$ is $q^u$-to-$1$ on $\mathbb{F}_{q^n}$.
\end{theorem}

\begin{proof}
Clearly, $L(x) = b(x^{q^{s-t}} - b^{-1}c x)\circ x^{q^t}$.
Thus it is $m$-to-$1$ on $\mathbb{F}_{q^n}$ if and only if
$L_r(x) \coloneqq x^{q^r} - a x$ is $m$-to-$1$ on $\mathbb{F}_{q^n}$, 
where $r = s - t$ and $a = b^{-1} c$.
By \cref{xnFq}, $x^{q^r-1}$ is $(q^u-1)$-to-$1$ on $\mathbb{F}_{q^n}^{*}$. 
Note that $L_r(x)$ is $m$-to-$1$ on $\mathbb{F}_{q^n}$ if and only if 
it has exactly~$m$ distinct roots in $\mathbb{F}_{q^n}$, 
or equivalently if and only if $x^{q^r-1} = a$ 
has exactly $m-1$ distinct solutions in $\mathbb{F}_{q^n}^{*}$. 
Hence $m-1 = 0$ if $a \notin (\mathbb{F}_{q^n}^{*})^{q^r-1}$   
and $m-1 = q^u-1$ if $a \in (\mathbb{F}_{q^n}^{*})^{q^r-1}$. 
That is, $m = 1$ if $a^{\frac{q^n-1}{q^u-1}} \neq 1$  
and $m = q^u$ if $a^{\frac{q^n-1}{q^u-1}} = 1$, 
since $(\mathbb{F}_{q^n}^{*})^{q^r-1} = (\mathbb{F}_{q^n}^{*})^{q^u-1}$, 
a cyclic group of order $(q^n-1)/(q^u-1)$.
\end{proof}

\begin{remark}
Let $L(x) = x^{q^r} - a x - b \in \mathbb{F}_{q^n}[x]$.
Coulter and Henderson {\cite{CoulterH04}}
determined the conditions under which 
$L(x)$ has roots in $\mathbb{F}_{q^n}$
and provided an expression of these roots.
Hence \cref{mto1_Lbi} can also be obtained 
by the method in {\cite[Theorem~3]{CoulterH04}}.  See also related work in Wu \cite{Wu-L-bi} and  Tuxanidy--Wang \cite{TuxanidyWang14}.
\end{remark}

\section{Main result}\label{sec_main}

We start with a lemma before we  state the main result.

\begin{lemma}\label{lem_axqbx}
Let $a$, $b$, $u$, $v \in\mathbb{F}_{q^2}$. 
Then the following statements hold:
\begin{enumerate}[\upshape(1)]
  \item The polynomial $ax^{q} +bx$ is $1$-to-$1$ on $\mathbb{F}_{q^2}$
        if and only if $a^{q+1} \neq b^{q+1}$.
  \item When $a b \neq 0$, $a^{q+1} = b^{q+1}$ if and only if $a = d e^q$ 
        and $b = d e$ for some $d$, $e \in \mathbb{F}_{q^2}^{*}$.
  \item There are $(q^2 - 1) (q + 1) q^2 (q^2 - 1)$ 
        choices for $a$, $b$, $u$, $v$ such that 
        $a^{q+1} = b^{q+1}$ and $a v \ne b u$.
\end{enumerate}   
\end{lemma}
\begin{proof}
Item (1) is a special case of \cref{mto1_Lbi}.
If $a^{q+1} = b^{q+1}$, then $(a/b)^{q+1} = 1$ 
and thus $a/b = e^{q-1}$ for some $e \in \mathbb{F}_{q^2}^{*}$.
Setting $d = b/e$ completes the proof of necessity for Item (2). 
The sufficiency of the condition in Item (2) is obvious. 
Finally,  the condition $a v \ne b u$ implies $a b \neq 0$.
Together with the fact that $x^{q+1}$ is $(q+1)$-to-$1$ 
on $\mathbb{F}_{q^2}^{*}$, we get the desired result for Item (3).
\end{proof}

From the above lemma, if $a^{q+1} \neq b^{q+1}$,
then $a x^q + b x$ permutes $\mathbb{F}_{q^2}$, and thus we can change the form of $f(x)$ given by \cref{eq_f_gen} using a change of variable $y=ax^q+bx$. 
If $a v = b u$, then $u x^q + v x = \alpha (a x^q + b x)$
for some $\alpha \in \mathbb{F}_{q^2}$, and thus $f(x) = h'(a x^q + b x + c) - \alpha c$, where $h'(x) = h(x) + \alpha x$. 
Hence, we assume that $a^{q+1} = b^{q+1}$ and $a v \ne b u$. 
Item (2) of \cref{lem_axqbx} implies that 
$a x^q + b x = d \tr(e x)$ if $a^{q+1} = b^{q+1} \neq 0$.
We are now ready to state the main result of this paper.

\begin{theorem}\label{thm_main_gen}
Let $a, b, c, u, v \in\mathbb{F}_{q^2}$ satisfy $a^{q+1} = b^{q+1}$ 
and $a v \ne b u$. Let $h \in \mathbb{F}_{q^2}[x]$~and
\begin{equation}\label{eq_f_gen}
  f(x) = h(a x^q + b x + c) + u x^q + v x.
\end{equation}
Assume $a = d e^q$ and $b = d e$ 
for some~$d$, $e \in \mathbb{F}_{q^2}^{*}$.
If $1 \le m \le q$, then $f(x)$ is $m$-to-$1$ on $\mathbb{F}_{q^2}$  
if and only if $m$ divides~$q$ and 
\begin{equation}\label{eq_g_gen}
  g(x) \coloneqq \tr \big( k h( d x + c) \big) + w x
\end{equation}
is $m$-to-$1$ on $\mathbb{F}_{q}$, where $k = e^q u^q - e v^q$
and $w = u^{q+1} - v^{q+1}$.
\end{theorem}

\begin{proof}
Since $d k^q = d e u - d e^q v = b u - a v \neq 0$, 
we get $k \neq 0$. For $x \in \mathbb{F}_{q^2}$, 
\begin{align*}
    \tr(k(u x^q + v x)) 
    & = (k u + k^q v^q) x^q + (k v + k^q u^q) x \\
    & = (u^{q+1} - v^{q+1})(e^q x^q + e x) \\
    & = w \tr(e x).
\end{align*}
Because $a = d e^q$ and $b = d e$, 
we have $a x^q + b x = d \tr(e x)$ and
\[
f(x) = h(d \tr(e x) + c) + u x^q + v x.
\]
Thus, for any $x \in \mathbb{F}_{q^2}$, 
\begin{equation}\label{eq_Lf=gL}
\begin{aligned}
\tr(k f(x)) 
& = \tr\big( k h(d \tr(e x) + c) \big) 
    + \tr \big(k (u x^q + v x) \big) \\
& = \tr\big( k h(d \tr(e x) + c) \big) + w \tr(e x) \\
& = g (\tr(e x)).
\end{aligned}
\end{equation}
Hence, the following diagram is commutative:
  \begin{equation*}
    \xymatrix{
    \mathbb{F}_{q^2} \ar[r]^{f} \ar[d]_{\lambda}   
    & \mathbb{F}_{q^2} \ar[d]^{\bar{\lambda}} \\
    \mathbb{F}_{q} \ar[r]^{g}       & \mathbb{F}_{q}
    }
  \end{equation*}
where $\lambda(x) = \tr(e x)$ and $\bar{\lambda}(x) = \tr(k x)$.

Clearly, $\#\lambda^{-1}(s) = q 
= \#\bar{\lambda}^{-1}(g(s))$ for any $s \in \mathbb{F}_{q}$.
Note that $f(x)$ can be rewritten as 
$f(x) = H(\lambda(x)) - \alpha x$, 
where $H(x) = h(d x + c) + (u / e^q) x$ 
and $\alpha = (e u / e^q) - v = k^q / e^q \neq 0$. 
Thus $f$ restricted to $\lambda^{-1}(s)$ 
is $1$-to-$1$ for each $s \in \mathbb{F}_{q}$.
Then by \cref{constr2}, for $1 \le m \le q$, 
$f$ is $m$-to-$1$ on $\mathbb{F}_{q^2}$  
if and only if~$g$ is $m$-to-$1$ on $\mathbb{F}_{q}$ and 
\begin{equation}\label{eq_qt=t2}
  q \cdot \# E_{g}(\mathbb{F}_{q}) = \# \mathbb{F}_{q^2}  \bmod{m}.
\end{equation}
Let $q = \ell m + t$ with $0 \le t < m \le q$. 
Then $q^2 = (\ell m)^2 + 2 t \ell m + t^2$ and  
thus \cref{eq_qt=t2} becomes $q t = t^2 \bmod{m}$. 
Hence $q t < m$. Since $q \geq m$, we have
$q t = t^2 \bmod{m}$ if and only if $t = 0$, 
or equivalently if and only if $m \mid q$.
The result then follows from \cref{constr2}. 
\end{proof}

This theorem reduces the problem of whether 
$f(x)$ is $m$-to-$1$ on $\mathbb{F}_{q^2}$ 
to the problem of whether~$g(x)$ is $m$-to-$1$ 
on the subfield $\mathbb{F}_{q}$ of $\mathbb{F}_{q^2}$.
Suppose $\sigma(x)$ is $1$-to-$1$ on $\mathbb{F}_{q}$. 
Then the composition $\sigma(g(x))$ is $m$-to-$1$ on $\mathbb{F}_{q}$ 
if and only if $g(x)$ is $m$-to-$1$ on $\mathbb{F}_{q}$. 
Hence, if we can determine the $m$-to-$1$ 
property of $\sigma(g(x))$ on $\mathbb{F}_{q}$, 
we can obtain the $m$-to-$1$ property of $f(x)$ on $\mathbb{F}_{q^2}$.
To illustrate \cref{thm_main_gen}, 
we will choose special classes of $h(x)$ such that 
each $\sigma(g(x))$ is congruent modulo $x^q - x$ to
a polynomial of degree $\leq 4$ or a linearized binomial, 
whose $m$-to-$1$ properties are fully classified in \cref{sec_mto1}. 
For simplicity, we only consider $h(x) = x^{r}$ 
for various exponents~$r$ in the forthcoming sections.
In this case, \cref{thm_main_gen} is reduced to the following form.

\begin{theorem}\label{thm_main}
Let $a$, $b$, $c$, $u$, $v \in\mathbb{F}_{q^2}$ satisfy 
$a^{q+1} = b^{q+1}$ and $a v \ne b u$. Let $r \in \mathbb{N}$ and
\begin{equation*}\label{eq_f}
  f_r(x) = (a x^q + b x + c)^{r} + u x^{q} + v x.
\end{equation*}
Assume $a = d e^q$ and $b = d e$ 
for some~$d$, $e \in \mathbb{F}_{q^2}^{*}$.
If $1 \le m \le q$, then $f(x)$ is $m$-to-$1$ on $\mathbb{F}_{q^2}$  
if and only if $m$ divides~$q$ and 
\begin{equation*}\label{eq_gr}
  g_{r}(x) \coloneqq \tr \big(k ( d x + c)^r \big) + w x
\end{equation*}
is $m$-to-$1$ on $\mathbb{F}_{q}$, where $k = e^q u^q - e v^q$
and $w = u^{q+1} - v^{q+1}$.
\end{theorem}

Since $f_{qr}(x)$ and $(f_r(x))^q$ 
share the same form as $f_r(x)$,  
it suffices to focus on $f_r(x)$. 
In Sections~\ref{sec_deg2}--\ref{sec_binomials}, 
we aim to explicitly characterize the $m$-to-$1$ property 
of $f_r(x)$ on $\mathbb{F}_{q^2}$ 
in terms of its coefficients $a$, $b$, $c$, $u$, and $v$.
To simplify notation, we write
\[
A = b u - a v  \qquad\text{and}\qquad 
B = a u^q - b v^q
\]
throughout the sequel. We shall also make frequent use of the identities
\[
k = B / d,  \qquad 
k^q = A / d, \qquad\text{and}\qquad  
d^{q-1} = a^q / b.
\]

\section{Degree two}\label{sec_deg2}

In this section, we study the cases  
$r=2$, $q+1$, $(q^2+q)/2$, or $q+t+1$, 
where~$t$ is a power of the characteristic of $\mathbb{F}_{q^2}$.
In each of these cases, $g_r(x)$ or $g_r^2(x)$ as defined in \cref{thm_main}
is congruent to a degree-$2$ polynomial modulo $x^q - x$,
and thus we can use \cref{deg2mto1} to characterize 
the $m$-to-$1$ property of $g_r(x)$ or $g_r^2(x)$ on $\mathbb{F}_{q}$.

\begin{theorem}\label{thm_r=2}
  Let $a$, $b$, $c$, $u$, $v \in\mathbb{F}_{q^2}$ satisfy 
  $a^{q+1} = b^{q+1}$ and $a v \ne b u$. Let  
  \[
    \alpha = a^q B^q + b B  \qquad\text{and}\qquad  
    \beta  = 2 B c + 2 B^q c^q + u^{q+1} - v^{q+1}.
  \]  
If $1 \le m \le q$, then $f_{2}(x)$ is $m$-to-$1$ on $\mathbb{F}_{q^2}$  if and only if 
one of the following holds:
\begin{enumerate}[\upshape(1)]
  \item $m = 1$, $\alpha = 0$, and $\beta \ne 0$; 
  \item $m = 1$, $\alpha \ne 0$, $\beta = 0$, and $q$ is even;
  \item $m = 2$, $\alpha \ne 0$, $\beta \ne 0$, and $q$ is even;
  \item $m = q$, $\alpha = 0$, and $\beta = 0$.
\end{enumerate}
\end{theorem}

\begin{proof}
For any $x \in \mathbb{F}_{q}$, we obtain
\begin{align*}
  g_{2}(x) & = \tr \big(k ( d x + c)^2 \big) + w x \\  
  & = \tr \big(k ( d^2 x^2 + 2 d c x + c^2) \big) + w x  \\
  & = \tr(k d^2) x^2 + \big(\tr(2 k d c) + w \big) x + \tr(k c^2).
\end{align*}
The identities $d e^q = a$ and $d e = b$ 
imply that $d^{q-1} = a^q / b$.  
Since $k d = B$, we have
\[
\tr(k d^2) = B d + B^q d^q 
= d (B + B^q a^q / b) 
= (d / b) (b B + B^q a^q) 
= e^{-1} \alpha.
\]
Thus, $g_{2}(x) = e^{-1} \alpha x^2 + \beta x + \tr(k c^2)$.   
We then get the result by
\cref{thm_main,deg2mto1}.
\end{proof}

\cref{thm_r=2} generalizes 
\cite[Proposition~3]{LiCao2389},
\cite[Proposition~1]{LLi18}, 
\cite[Theorem~3.2]{PangY27Fq2},
and \cite[Theorem~4.6]{XuLC22},   
where $m = a = -b = 1$ and $u \in  \{0, 1\}$.

\begin{theorem}\label{thm_r=q+1}
Let $a$, $b$, $c$, $u$, $v \in\mathbb{F}_{q^2}$ satisfy
$a^{q+1} = b^{q+1}$ and $a v \ne b u$. Let  
\[ 
  \alpha = A+B   \qquad\text{and}\qquad 
  \beta = \alpha (a^q c + b c^q) + b (u^{q+1} - v^{q+1}).
\]
If $1 \le m \le q$, then $f_{q + 1}(x)$ 
is $m$-to-$1$ on $\mathbb{F}_{q^2}$  if and only if one of the following holds:
\begin{enumerate}[\upshape(1)]
  \item $m = 1$, $\alpha = 0$, and $u^{q+1} \ne v^{q+1}$;
  \item $m = 1$, $\alpha \ne 0$, $\beta = 0$, and $q$ is even;
  \item $m = 2$, $\alpha \ne 0$, $\beta \ne 0$, and $q$ is even;
  \item $m = q$, $\alpha = 0$, and $u^{q+1} = v^{q+1}$.
\end{enumerate}
\end{theorem}

\begin{proof}
When $r = q + 1$, for any $x \in \mathbb{F}_{q}$ we have
\begin{align*}
  g_r(x) 
  & = \tr \big(k (d x + c)^{q + 1} \big) + w x \\
  & = \tr(k) (d^{q + 1} x^2 + (c d^q + c^q d) x + c^{q + 1})  + w x \\
  & = \tr(k) d^{q + 1} x^2 + \big( \tr(k) (c d^q + c^q d) + w\big) x + \tr(k) c^{q + 1}.  
\end{align*}
Since $k = B / d$ and $k^q = A / d$, 
we get $\tr(k) = (A + B) / d$. 
Using $d^{q-1} = a^q / b$, we have 
\begin{equation}\label{eq_cdq}
c d^q + c^q d = d (c a^q / b + c^q) 
= (d/b) (a^q c + b c^q).
\end{equation}
Hence, $g_r(x) = d^q \alpha x^2 +  b^{-1} \beta x + \tr(k) c^{q + 1}$.
The result then follows from \cref{thm_main,deg2mto1}.  
\end{proof}

\cref{thm_r=q+1} generalizes \cite[Theorem~2]{LiuJZ25Fq2}
(where $q = 3^n$, $m = a = - b = 1$,  and $u = v$)
and \cite[Proposition~4.13]{WuY24Fq2} 
(where $m = a = 1$, $u = 0$, and $q$ is even).

\begin{theorem}\label{thm_q2+qdiv2}
Let $q$ be even and $a$, $b$, $c$, $u$, $v \in\mathbb{F}_{q^2}$
satisfy $a^{q+1} = b^{q+1}$ and $a v \ne b u$.~Let 
\[ 
  \alpha = a^q (A+B)^2 + b(u^{q+1} + v^{q+1})^2 
  \qquad\text{and}\qquad 
  \beta = (A+B)^2 (a^q c + b c^q).
\]
If $r = (q^2+q)/2$ and $1 \le m \le q$, 
then $f_{r}(x)$ is $m$-to-$1$ on $\mathbb{F}_{q^2}$  if and only if one of the following holds:
\begin{enumerate}[\upshape(1)]
  \item $m = 1$, $\alpha = 0$, and $\beta \neq 0$;
  \item $m = 1$, $\alpha \neq 0$, and $\beta = 0$;
  \item $m = 2$, $\alpha \ne 0$, and $\beta \ne 0$;
  \item $m = q$, $\alpha = 0$, and $\beta = 0$.
\end{enumerate}
\end{theorem}
\begin{proof}
Since $x^2$ permutes $\mathbb{F}_{q}$, $g_r(x)$ is $m$-to-$1$ on $\mathbb{F}_{q}$ 
if and only if so is $g_r^2(x)$. For $x \in \mathbb{F}_{q}$,  
\begin{align*}
  g_r^2(x) 
  & = \tr \big( k^2 (d x + c)^{q^2 + q} \big) + w^2 x^2 \\
  & = \tr \big( k^2 (d x + c)^{1 + q} \big) + w^2 x^2 \\
  & = \tr(k^2) (d^{q + 1} x^2 + (c d^q + c^q d) x + c^{q + 1}) + w^2 x^2 \\
  & = \big( \tr(k^2) d^{q + 1} + w^2 \big) x^2 + \tr(k^2) (c d^q + c^q d)  x + \tr(k^2) c^{q + 1}\\
  & = b^{-1} \alpha x^2 + (bd)^{-1} \beta x + \tr(k^2)  c^{q + 1},
\end{align*}
where $\tr(k^2) = \tr(k)^2 = (A + B)^2 / d^2$,
$d^{q-1} = a^q / b$, and \cref{eq_cdq} 
are used in the last~identity.
The result then follows from \cref{thm_main,deg2mto1}.
\end{proof}

\cref{thm_q2+qdiv2} generalizes 
\cite[Theorem~3.2~(1)]{WangWuL17} and \cite[Theorem~4.15]{WuY24Fq2}
where $m = a = 1$ and $u=0$. 
It also generalizes \cite[Proposition~10]{XuFZ19}
where $m =  a = b = 1$ and $c^q \ne c$.

\begin{theorem}\label{r=q+t+1}
Let $a$, $b$, $c$, $u$, $v \in\mathbb{F}_{q^2}$ satisfy $a^{q+1} = b^{q+1}$, 
$a v \ne b u$,  and $(a^q/b)^t A = - B$, where~$t$ is a power of 
the characteristic of $\mathbb{F}_{q^2}$. Let  
\[
  \alpha = a^q c - b c^q \qquad\text{and}\qquad 
  \beta  = B \alpha^t (a^q c + b c^q)  + a^{qt} b (u^{q+1} - v^{q+1}).
\]
If $r = q + t + 1$ and $1 \le m \le q$, 
then $f_{r}(x)$ is $m$-to-$1$ on $\mathbb{F}_{q^2}$  if and only if one of the following holds:
\begin{enumerate}[\upshape(1)]
  \item $m = 1$, $\alpha = 0$, and $u^{q+1} \ne v^{q+1}$;
  \item $m = 1$, $\alpha \ne 0$, $\beta = 0$, and $q$ is even;  
  \item $m = 2$, $\alpha \ne 0$, $\beta \ne 0$, and $q$ is even;  
  \item $m = q$, $\alpha = 0$, and $u^{q+1} = v^{q+1}$.
\end{enumerate}
\end{theorem}

\begin{proof}
When $r = q + t + 1$, for any $x \in \mathbb{F}_{q}$ we obtain
\begin{equation}\label{eq_gr_qt1}
\begin{aligned}
  g_r(x) 
  & = \tr \big(k (d x + c)^{q + t + 1} \big) + w x \\
  & = \tr \big(k (d x + c)^{t} \big) (d x + c)^{q + 1} + w x \\
  & = \big(\tr(k d^{t}) x^{t} + \tr(k c^{t})\big) (d^{q + 1} x^2 + (c d^q + c^q d) x + c^{q + 1}) + w x. 
\end{aligned}
\end{equation}
Since $k = B / d$, $k^q = A / d$, $d^{q-1} = a^q / b$, 
and $(a^q/b)^t A = - B$, we have 
\begin{equation}\label{eq_trkdt}
\begin{aligned}
\tr(k d^{t}) 
& = k d^{t} + k^q d^{qt}
= (B d^{t} + A d^{qt})/d 
= d^{t-1} (B + A (a^q / b)^t) 
= 0, \\
\tr(k c^{t}) 
& = (B c^{t} + A c^{qt})/d 
= (B c^{t} - B (b / a^q)^t c^{qt})/d 
= (a^q c - b c^q)^t B / (a^{qt} d).
\end{aligned}
\end{equation}
Substituting \eqref{eq_cdq} and \eqref{eq_trkdt} into \eqref{eq_gr_qt1} yields 
\[
g_r(x) = \alpha^t B a^{-qt} d^q x^2 + (a^{qt} b)^{-1} \beta x + \tr(k c^{t}) c^{q + 1}.
\]
Note that $a b d B \neq 0$.
The result then follows from \cref{thm_main,deg2mto1}. 
\end{proof}

\cref{r=q+t+1} generalizes \cite[Theorem~4.5]{WuY24Fq2} 
and \cite[Theorem~5.9]{2to1-YuanZW21}, 
where $2 \mid q$, $m = 1$ or $2$, $a = b = 1$, 
$u = 0$, and $v \in \mathbb{F}_{q}^{*}$.

\section{Degree three}\label{sec_deg3}

This section is concerned with the cases where 
$r=3$, $q+2$, or $(2q^2+q)/3$.
In each of these cases, $g_r(x)$ or $g_r^3(x)$ 
is congruent to a degree-$3$ polynomial modulo $x^q - x$, 
and so we can apply the classification results from  \cref{sec_mto1}
to characterize the $m$-to-$1$ property of 
$g_r(x)$ or~$g_r^3(x)$.

\begin{theorem}\label{thm_r=3}
Let $q \geq 7$ and $a$, $b$, $c$, $u$, $v \in\mathbb{F}_{q^2}$ satisfy 
$a^{q+1} = b^{q+1}$ and $a v \ne b u$. Let  
\begin{align*}
  \alpha = a^{2q} B^q + b^2 B,  \quad
  \beta  = 3 (a^q c - b c^q),       \quad
  \gamma = v^{q+1} - u^{q+1},   \quad
  \delta = \beta^2 B - a^{2q} \gamma. 
\end{align*}
If $1 \le m \le q$, then $f_{3}(x)$ is $m$-to-$1$ on $\mathbb{F}_{q^2}$  if and only if 
one of the following holds:
\begin{enumerate}[\upshape(1)]
  \item $m = 1$, $\alpha = 0$, $\beta = 0$, and $\gamma \ne 0$;
  \item $m = 1$, $2 \mid q$, $\alpha = 0$, $\beta \neq 0$, and $\delta = 0$;
  \item $m = 1$, $3 \mid q$, $\alpha \neq 0$, and
        $(\alpha \gamma)^{\frac{q-1}{2}} \neq a b^{-1}$;
  \item $m = 1$, $q \equiv 2 \pmod{3}$, $\alpha \neq 0$,  and 
        $\beta^2 B^{q+1} = 3 \alpha \gamma$;
  \item $m = 2$, $2 \mid q$, $\alpha = 0$,  $\beta \ne 0$, 
        and $\delta \ne 0$;
  \item $m = 3$, $3 \mid q$, $\alpha \neq 0$, and 
        $(\alpha \gamma)^{\frac{q-1}{2}} = a b^{-1}$;
  \item $m = q$, $\alpha = 0$, $\beta = 0$, and $\gamma = 0$.
\end{enumerate}
\end{theorem}

\begin{proof}
When $r = 3$, for any $x \in \mathbb{F}_{q}$ we obtain
\begin{align*}
  g_r(x) 
  & = \tr \big(k (d x + c)^{3} \big) + w x \\
  & = \tr \big(k (d^3 x^3 + 3 c d^2 x^2 + 3 c^2 d x + c^3) \big) + w x \\ 
  & = \tr(k d^3) x^3 + 3 \tr(k c d^2) x^2 
        + \big( 3 \tr(k c^2 d) + w \big) x + \tr(k c^3). 
\end{align*}
Since $k d = B$, $d^{q-1} = a^q / b$ and $d e = b$, 
we have $\tr(k c^2 d) = \tr(B c^2)$, 
\begin{align*}
& \tr(k d^3) = B d^2 + B^q d^{2q} 
= d^2 (B + B^q (a^q / b)^2) = e^{-2} \alpha, \\
& \tr(k c d^2) = B c d + (B c d)^q 
= d (B c + B^q c^q a^q / b ) = e^{-1} (b c B + (a c B)^q).
\end{align*}
Hence, for any $x \in \mathbb{F}_{q}$,
$g_r(x) = e^{-2} \alpha x^3 + e^{-1} \beta' x^2 
        + \gamma' x + \tr(k c^3)$, where
\[
  \beta'  = 3bcB + 3(acB)^q  \quad\text{and}\quad
  \gamma' = 3Bc^2 + 3(Bc^2)^q + u^{q+1} - v^{q+1}. 
\]
Since $d e^q = a$ and $d e = b$, we have 
$(e^{-2})^{\frac{q-1}{2}} = e^{-(q-1)} = b/a$. 
It is easy to verify that
\[
\begin{aligned}
(\beta')^2 + \beta^2 B^{q+1}
= 9 \alpha (Bc^2 + B^q c^{2q})
= 3 \alpha (\gamma' + \gamma).
\end{aligned}
\]
Thus $(\beta')^2 = 3 \alpha \gamma'$ if and only if 
$\beta^2 B^{q+1} = 3 \alpha \gamma$. 

If $\alpha = 0$, then $B^q = - a^{-2q} b^2 B$, and so
$\beta' = a^{-q} b B \beta$ and
$\gamma' =  a^{-2q} B (a^q c + b c^q) \beta - \gamma$.
Since $b B \ne 0$, we get $\beta' = 0$ if and only if $\beta = 0$.
By \cref{thm_main,deg2mto1}, 
$f_{r}(x)$ is $m$-to-$1$ on $\mathbb{F}_{q^2}$  if and only if one of 
(1), (2), (5), and (7) holds. 

If $\alpha \neq 0$, by Theorems \ref{deg3-1to1}, 
\ref{deg3-2to1}, \ref{deg3-3to1} and \ref{thm_main}, 
$f_{r}(x)$ is $m$-to-$1$ on $\mathbb{F}_{q^2}$  if and only if one of 
(3), (4), and (6) holds. 
\end{proof}

\cref{thm_r=3} generalizes \cite[Proposition~4.28]{nto1-NiuLQL23} 
where $m=2$, $a = b = v = 1$~and~$u=0$.

\begin{theorem}\label{thm_r=q+2}
 Let $q \geq 7$ and $a$, $b$, $c$, $u$, $v \in\mathbb{F}_{q^2}$ satisfy 
 $a^{q+1} = b^{q+1}$ and $a v \ne b u$. Let  
\begin{align*}
    \alpha & = B + B^q, \\
    \beta  & = \alpha (a^q c + b c^q) + a^q c B + b c^q B^q, \\
    \gamma & = 2 \alpha c^{q+1} + A c^{2q} + A^q c^2 + u^{q+1} - v^{q+1}. 
\end{align*}
If $1 \le m \le q$, 
then $f_{q + 2}(x)$ is $m$-to-$1$ on $\mathbb{F}_{q^2}$  if and only if 
one of the following holds:
\begin{enumerate}[\upshape(1)]
  \item $m = 1$, $\alpha = 0$, $a^q c = b c^q$, and $u^{q+1} \ne v^{q+1}$;
  \item $m = 1$, $2 \mid q$, $\alpha = 0$, $a^q c \neq b c^q$, and $\gamma = 0$;  
  \item $m = 1$, $q \equiv 2 \pmod{3}$, $\alpha \ne 0$, 
            and $\beta^2 = 3 a^q b \alpha \gamma$;
  \item $m = 1$, $3 \mid q$, $\alpha \ne 0$, $\beta = 0$, and 
        $(-\alpha \gamma)^{\frac{q-1}{2}} \ne (a b^{-q})^{\frac{q+1}{2}}$;
  \item $m = 2$, $2 \mid q$, $\alpha = 0$,  $a^q c \neq b c^q$, and $\gamma \ne 0$;
  \item $m = 3$, $3 \mid q$,  $\alpha \ne 0$,  $\beta = 0$, 
            and $(-\alpha \gamma)^{\frac{q-1}{2}} = (a b^{-q})^{\frac{q+1}{2}}$;
  \item $m = q$, $\alpha = 0$, $a^q c = b c^q$, and $u^{q+1} = v^{q+1}$.
\end{enumerate}
\end{theorem}

\begin{proof}
When $r = q + 2$, for any $x \in \mathbb{F}_{q}$, we have  
\begin{equation}\label{eq_gr_q2}
\begin{aligned}
  g_r(x) 
  & = \tr \big(k (d x + c)^{q + 2} \big) + w x \\
  & = \tr \big(k (d x + c) \big) (d x + c)^{q + 1} + w x \\
  & = \big(\tr(k d) x + \tr(k c)\big) (d^{q + 1} x^2 
        + (c d^q +  c^q d) x + c^{q + 1}) + w x. 
\end{aligned}
\end{equation}
Since $k = B / d$, $d^{q-1} = a^q / b$, and \cref{eq_cdq} holds, 
it follows that $\tr(k d) = \tr(B)$ and
\begin{equation}\label{eq_trkdkc}
\begin{gathered}
\tr(k d) (c d^q +  c^q d) = \tr(B) (d / b) (a^q c + b c^q), \\
\tr(k c) = (B/d) c + (B/d)^q c^q
= (a^q c B + b c^q B^q) / (a^q d), \\
\tr(k c)(c d^q +  c^q d)
= \tr(k d^q c^2)  + \tr(k d c^{q+1})
= \tr(A^q c^2)  + \tr(B c^{q+1}).
\end{gathered}
\end{equation}
Therefore, $g_r(x) = d^{q+1} \alpha x^3 
+ (d / b) \beta x^2 + \gamma x +  \tr(k c) c^{q+1}$. 
If $\alpha = 0$, then $\beta = (a^q c - b c^q) B$. 
If $\alpha = \beta = 0$, then $\tr(k d) = \tr(k c) = 0$ 
by \cref{eq_trkdkc}, 
and thus $g_r(x) = w x$ by \cref{eq_gr_q2}, 
where $w = u^{q+1} - v^{q+1}$. 
The remainder of the proof is similar to that used in \cref{thm_r=3}.
\end{proof}

\cref{thm_r=q+2} generalizes and unifies 
\cite[Theorem~3~(1)]{TuZJ1531}, 
\cite[Theorem~2]{TuZLH1534},
\cite[Corollaries~1 and~5]{WuY22Fq2}, 
\cite[Theorem~3.6]{WuY24Fq2}, 
and \cite[Theorems~3.1, 3.13 and~3.18]{PangY27Fq2},
where $m = 1$ and $u = 0$.
\cref{thm_r=q+2} also generalizes \cite[Proposition~1]{LiCao2389} 
(where $m=1$, $2 \mid q$, $a = b = 1$, 
and $u + v \in \mathbb{F}_{q}^{*}$)
and \cite[Proposition~4.24]{nto1-NiuLQL23} 
(where $m= 3$, $- a = b = v = 1$, and $u = 0$).

\begin{theorem}\label{thm_r=(2q+3)/3}
Let $q = 3^n \geq 9$ and $a, b, c, u, v \in\mathbb{F}_{q^2}$ 
satisfy $a^{q+1} = b^{q+1}$ and $a v \ne b u$.~Let 
\begin{align*}
  \alpha &= a^q(a^q A^3 + b B^3) + b^2(u^{q+1} - v^{q+1})^3, \\
  \beta  &= (a^q c - b c^q) (a^q A^3 - b B^3), \\
  \gamma &= (a^q c - b c^q) (B^3 c- A^3 c^q).
\end{align*}
If $r = (2q^2+q)/3$ and $1 \le m \le q$, then $f_{r}(x)$ is $m$-to-$1$ on $\mathbb{F}_{q^2}$  if and only if 
one of the following holds:
\begin{enumerate}[\upshape(1)]
  \item $m = 1$, $\alpha = 0$, $\beta = 0$, and $\gamma \neq 0$;
  \item $m = 1$, $\alpha \neq 0$, $\beta = 0$, and 
$(-\alpha \gamma)^{(q-1)/2} \neq a^q b^{(3q-5)/2}$;
  \item $m = 3$, $\alpha \ne 0$, $\beta = 0$, and $(-\alpha \gamma)^{(q-1)/2} = a^q b^{(3q-5)/2}$;
  \item $m = q$, $\alpha = 0$, $\beta = 0$, and $\gamma = 0$.
\end{enumerate}
\end{theorem}
\begin{proof}
Since $x^3$ permutes $\mathbb{F}_{q}$, $g_r(x)$ 
is $m$-to-$1$ on $\mathbb{F}_{q}$ if and only if so is $g_r^3(x)$.
Then for any $x \in \mathbb{F}_{q}$, a routine calculation yields
\[
  g_r^3(x) 
  = b^{-2} \alpha x^3 + b^{-2} d^{-1} \beta x^2 
  + b^{-1} d^{-2} \gamma x + \tr(k^3 c^{q+2}).
\]
The remainder of the proof is similar to that used in \cref{thm_r=3}.
\end{proof}

\cref{thm_r=(2q+3)/3} generalizes
\cite[Theorem~3.7]{Gupta18}, 
\cite[Proposition~3]{LLi18}, 
\cite[Theorems~4.1 and~4.2~(4)]{WangWuL17}, 
\cite[Corollary~6]{WuY22Fq2},
and \cite[Proposition~1]{DZheng17}, 
where $m = 1$, $a = - b = \pm 1$, 
$v = 1$, and $u = 0$ or~$1$.

\section{Degree four}\label{sec_deg4}

This section deals with the cases   
$r = 4$ or $(q^2 + q + 2)/2$, where~$q$ is even.
In each of these cases, $g_r(x)$ or $g_r^2(x)$ 
is congruent to a degree-$4$ polynomial modulo $x^q - x$,
and thus we can apply the classification results from  \cref{sec_mto1} 
to characterize the $m$-to-$1$ property of $g_r(x)$ or $g_r^2(x)$ on $\mathbb{F}_{q}$. 

\begin{theorem}\label{thm_r=4}
Let $q = 2^n \geq 4$ and $a, b, c, u, v \in\mathbb{F}_{q^2}$ satisfy
$a^{q+1} = b^{q+1}$ and $a v \ne b u$.~Let
\begin{align*}
    \alpha & = a^{3q} B^{q} + b^3 B 
    \qquad\text{and}\qquad 
    \beta = u^{q+1} + v^{q+1}.
\end{align*}
If $1 \le m \le q$, then $f_{4}(x)$ is $m$-to-$1$ on $\mathbb{F}_{q^2}$ 
if and only if one of the following holds:
\begin{enumerate}[\upshape(1)]
  \item $m = 1$, $\alpha = 0$, and $\beta \ne 0$;
  \item $m = 1$, $\alpha \ne 0$, and $\beta = 0$;
  \item $m = 1$, $\alpha \beta \neq 0$, $n$ is even, 
        and $(\beta / \alpha)^{\frac{q-1}{3}} \neq b / a$.
  \item $m = 2$, $\alpha \beta \ne 0$, and $n$ is odd;
  \item $m = 4$, $\alpha \beta \ne 0$, $n$ is even, 
        and $(\beta / \alpha)^{\frac{q-1}{3}} = b / a$;
  \item $m = q$, $\alpha = 0$, and $\beta = 0$.
\end{enumerate}
\end{theorem}

\begin{proof}
Since $k d = B$, $d^{q-1} = a^q / b$ and $d e = b$, we have 
\[
\tr(k d^4) 
= B d^3 + B^q d^{3q}
= d^3 (B + B^q (a^q / b)^3)
= (d/b)^3 (b^3 B + a^{3q} B^q)
= e^{-3} \alpha.
\]
Thus, for $x \in \mathbb{F}_{q}$, 
$g_{4}(x) = e^{-3} \alpha x^4 + \beta x + \tr(k c^4)$.
If $\alpha \beta \neq 0$, then 
$\mathrm{Tr}_{2^n / 2^{(2, n)}} (e^3 \beta / \alpha) = 1$
if and only if (i)~$n$ is odd, or (ii)~$n$ is even and 
$(\beta / \alpha)^{\frac{q-1}{3}} = e^{-(q-1)}$. 
The result then follows from \cref{thm_main,mto1_Lbi}. 
\end{proof}

We next consider the case $r = (q^2 + q + 2)/2$, 
for which we need the following lemma. 

\begin{lemma}[{\cite[$\S$71]{Dickson1896}}]\label{x4a1=0}
  Let $f(x) = x^4 + a_3 x^3 + a_2 x^2 + a_1 x \in \mathbb{F}_{2^n}[x]$ with $n \geq 3$. 
  If $f(x)$ is $1$-to-$1$ on $\mathbb{F}_{2^n}$, then $a_3 = 0$.
\end{lemma}


\begin{theorem}\label{thm_r=(q+1)q/2+1}
Let $q = 2^n \geq 4$ and $a, b, c, u, v \in\mathbb{F}_{q^2}$ satisfy
$a^{q+1} = b^{q+1}$ and $a v \ne b u$.~Let
\begin{align*}
    \alpha & = B + B^q, \quad
    \beta = a^q c + b c^q, \quad
    \gamma = u^{q+1} + v^{q+1}, \\
    \delta &= c^{q+1} \alpha^2 + a^{q}b^{-1}(A c^q + B c)^2
            + \gamma^2.
\end{align*}
If $r = (q^2 + q + 2)/2$, then $f_{r}(x)$ 
is $1$-to-$1$ on $\mathbb{F}_{q^2}$ if and only if one of the following holds:
\begin{enumerate}[\upshape(1)]
  \item $\alpha = 0$, $\beta = 0$,   and $\gamma \ne 0$;
  \item $\alpha = 0$, $\beta \ne 0$, and $\delta = 0$;
  \item $\alpha \ne 0$, $\beta = 0$, and $\gamma = 0$.
\end{enumerate}
If $r = (q^2 + q + 2)/2$, then $f_r(x)$ 
is $2$-to-$1$ on $\mathbb{F}_{q^2}$ 
if and only if one of the following holds:
\begin{enumerate}[resume*]
  \item $\alpha = 0$, $\beta \ne 0$, and $\delta \ne 0$;
  \item $\alpha \ne 0$, $\beta = 0$, and $\gamma \ne 0$;
  \item $\alpha \ne 0$, $\beta \ne 0$, 
        $b\delta = \alpha \beta (Ac^q + Bc)$, and $n$ is odd.
\end{enumerate}
\end{theorem}

\begin{proof}
Since $x^2$ permutes $\mathbb{F}_{q}$, $g_{r}(x)$ 
is $m$-to-$1$ on $\mathbb{F}_{q}$ if and only if $g_{r}^2(x)$ is $m$-to-$1$ on $\mathbb{F}_{q}$.
For any $x \in \mathbb{F}_{q}$, a routine calculation yields that
$g_r^2(x) = \sum_{i = 0}^{4} e_i x^i$,
where 
\begin{align*}
  e_4 &= d^{q+1} (B + B^q)^2,  \\
  e_3 &= b^{-1} d (a^q c + b c^q) (B + B^q)^2, \\
  e_2  &= c^{q+1} (B + B^q)^2 + a^{q}b^{-1}(A c^q + B c)^2
            + (u^{q+1} - v^{q+1})^2, \\
  e_1 &= (b d)^{-1} (a^q c + b c^q) (A c^q + B c)^2, \\
  e_0 &= d^{-2} (B^2 c^{q+3} + A^2 c^{3q+1}).
\end{align*}
It is easy to verify that $b (A c^q - B^q c) = A (b c^q -a^q c)$.
Hence $A c^q = B^q c$ if and only if $a^q c = b c^q$.
In particular, if $B^q = B$, then $e_1 = 0$ if and only if $a^q c = b c^q$. 
Next, we consider the $1$-to-$1$ property of $g_r^2(x)$.

If $B^q = B$, then $e_4 = e_3 = 0$, and so 
$g_r^2(x) = e_2 x^2 +  e_1 x + e_0$.
By \cref{deg2pps}, $g_r^2(x)$ is $1$-to-$1$ on $\mathbb{F}_{q}$ 
if and only if exactly one of $e_1$ and $e_2$ is~$0$, 
or equivalently if and only if (1) or (2) holds.

If $B^q \ne B$ and $a^q c = b c^q$, 
then $e_4 \ne 0$, $e_3 = e_1 = 0$, and 
\[
e_2 = c^{q+1} (B + B^q)^2 + a^{q}b^{-1}(B^q c + B c)^2            
    + (u^{q+1} - v^{q+1})^2    
= (u^{q+1} - v^{q+1})^2.
\]
Hence $g_r^2(x) - e_0 = e_4 x^4 + e_2 x^2 
= (e_4 x^2 + (u^{q+1} - v^{q+1})^2 x) \circ x^2$. 
Thus $g_r^2(x)$ is $1$-to-$1$ on $\mathbb{F}_{q}$ if and only if 
$e_4 x^2 + (u^{q+1} - v^{q+1})^2 x$ is $1$-to-$1$ on $\mathbb{F}_{q}$, 
or equivalently if and only if $u^{q+1} = v^{q+1}$.

If $B^q \neq B$ and $a^q c \ne b c^q$, 
then $e_4 \ne 0$ and $e_3 \ne 0$. 
When $q \geq 8$, $g_r^2(x)$ is not 
$1$-to-$1$ on $\mathbb{F}_{q}$ by \cref{x4a1=0}. 
When $q=4$, $g_r^2(x) = e_3 x^3 + e_2 x^2 + (e_4 + e_1) x + e_0$ 
for $x \in \mathbb{F}_{4}$, which is not $1$-to-$1$ on $\mathbb{F}_{4}$ by \cref{deg3-1to1}.  
This completes the proof for $1$-to-$1$ mappings. 

The proof of the $2$-to-$1$ mappings  
follows from \cref{deg2-2to1,deg4-2to1-p=2}.
\end{proof}

\cref{thm_r=(q+1)q/2+1} generalizes and unifies 
\cite[Theorem~3~(2)]{TuZJ1531}
(where $m =  a = b = v = 1$, $u = 0$, and $c + c^q = 0$ or~$1$) and
\cite[Corollary~2]{WuY22Fq2} (where $m = a = 1$ and $u=0$).

\section{Linearized binomials}\label{sec_binomials}

This section treats the cases where either (i) $r = p^s$,  
or (ii) $r = p^s + 1$ under the additional assumption $(a^q/b)^{r} A = - B$. In each of these cases, 
$g_r(x) \equiv \alpha x^{p^s} + \beta x + \gamma \pmod{x^q - x}$,
and thus we can use \cref{mto1_Lbi} 
to characterize the $m$-to-$1$ property of $g_r(x)$ on $\mathbb{F}_{q}$.

\begin{theorem}\label{thm_r_ps}
Let $a$, $b$, $c$, $u$, $v \in\mathbb{F}_{q^2}$ satisfy 
$a^{q+1} = b^{q+1}$ and $a v \ne b u$. 
Let $q = p^n$, $r = p^s$, and $t = p^{(s, n)} - 1$, 
where $n \geq 1$, $s \geq 0$, 
and~$p$ is the characteristic of $\mathbb{F}_{q^2}$. Write  
\[
  \alpha = (a^q/b)^{r} A + B
  \qquad\text{and}\qquad 
  \beta = u^{q+1} - v^{q+1}.
\]
If $1 \le m \le q$, then $f_{r}(x)$ is $m$-to-$1$ on $\mathbb{F}_{q^2}$  if and only if 
one of the following holds:
\begin{enumerate}[\upshape(1)]
  \item $m = 1$, $\alpha = 0$, and $\beta \ne 0$;
  \item $m = 1$, $\alpha \neq 0$, and 
    $(-\beta/\alpha)^{\frac{q-1}{t}}
    \neq (a^q/b)^{\frac{p^s-1}{t}}$;
  \item $m = p^{(s, n)}$, $\alpha \neq 0$, and 
    $(-\beta/\alpha)^{\frac{q-1}{t}}
    = (a^q/b)^{\frac{p^s-1}{t}}$;
  \item $m = q$, $\alpha = 0$, and $\beta = 0$.
\end{enumerate}
\end{theorem}

\begin{proof}
For any $x \in \mathbb{F}_{q}$, we have
$g_r(x) = d^{r-1} \alpha x^{r} + \beta x + \tr(k c^{r})$.   
Note that 
\[
\mathrm{Tr}_{p^n/p^{(s,n)}} \Big(\frac{-\beta}{d^{r-1} \alpha}\Big) 
= \Big(\frac{-\beta}{d^{r-1}     \alpha}\Big)^{\frac{q-1}{t}} 
= \Big(\frac{-\beta}{\alpha}\Big)^{\frac{q-1}{t}} 
  d^{(q-1) \frac{1-r}{t}} 
= \Big(\frac{-\beta}{\alpha}\Big)^{\frac{q-1}{t}}   
  \Big(\frac{b}{a^q}\Big)^{\frac{r-1}{t}}.
\]
The result then follows from \cref{thm_main,mto1_Lbi}. 
\end{proof}

\cref{thm_r_ps} generalizes \cite[Theorem~4.12]{WuY24Fq2} 
where $m = a = 1$, $u = 0$, and $2 \mid q$.
It also recovers \cref{thm_r=4} 
when~$q$ is even and $r = 4$.

\begin{theorem}\label{thm_r_ps1}
Let $a$, $b$, $c$, $u$, $v \in \mathbb{F}_{q^2}$ satisfy $a^{q+1} = b^{q+1}$ and $a v \ne b u$. 
Let $q = p^n$, $r = p^s + 1$, $t = p^{(s, n)} - 1$, and $(a^q/b)^{r} A = - B$, 
where $n \geq 1$, $s \geq 0$, and~$p$ is the characteristic of $\mathbb{F}_{q^2}$. Write 
\[
  \alpha = B (c - a^{-q} b c^q)  \qquad\text{and}\qquad 
  \beta  = B (c - a^{-q}  b c^q)^{p^s} + u^{q+1} - v^{q+1}.
\]
If $1 \le m \le q$, then $f_{r}(x)$ is $m$-to-$1$ on $\mathbb{F}_{q^2}$  if and only if 
one of the following holds:
\begin{enumerate}[\upshape(1)]
  \item $m = 1$, $\alpha = 0$, and $u^{q+1} \neq v^{q+1}$;
  \item $m = 1$, $\alpha \neq 0$, and 
    $(-\beta/\alpha)^{\frac{q-1}{t}}
    \neq (a^q/b)^{\frac{p^s-1}{t}}$;
  \item $m = p^{(s, n)}$, $\alpha \neq 0$, and 
    $(-\beta/\alpha)^{\frac{q-1}{t}}
    = (a^q/b)^{\frac{p^s-1}{t}}$;
  \item $m = q$, $\alpha = 0$, and $u^{q+1} = v^{q+1}$.
\end{enumerate}
\end{theorem}

\begin{proof}
When $r = p^s + 1$, for any $x \in \mathbb{F}_{q}$ we obtain
\begin{equation}\label{eq_gr_ps1}
\begin{aligned}
  g_r(x) 
  & = \tr \big(k (d x + c)^{r} \big) + w x \\
  & = \tr \big(k ((d x)^{r} + c (d x)^{p^s} 
        + c^{p^s} d x + c^{r}) \big) + w x \\  
  & = \tr(k d^{r}) x^{r} 
      + \tr(k c d^{p^s}) x^{p^s} 
      + ( \tr(k c^{p^s} d) + w ) x 
      + \tr(k c^{r}). 
\end{aligned}
\end{equation}
Since $k = B / d$, $k^q = A / d$, $d^{q-1} = a^q / b$, 
and $(a^q/b)^{r} A = - B$, we have 
\begin{equation}\label{eq_trkdps1}
\begin{aligned}
& \tr(k d^{r}) 
= k d^{r} + k^q d^{q r}
= (B d^{r} + A d^{q r})/d 
= d^{p^s} (B + A (a^q / b)^r) 
= 0, \\
& \tr(k c d^{p^s}) 
= (B c d^{p^s} + A c^q d^{p^s q})/d 
= d^{p^s - 1} (B c + A c^q (a^q/b)^{p^s}) 
= d^{p^s - 1} \alpha. \\
& \tr(k c^{p^s} d) 
= (B c^{p^s} d + A c^{p^s q} d^q)/d
= B c^{p^s} + A (a^q / b) c^{p^s q}
= B (c - a^{-q} b c^q)^{p^s}.
\end{aligned}
\end{equation}
Substituting \eqref{eq_trkdps1} into \eqref{eq_gr_ps1} yields 
$g_r(x) = d^{p^s - 1} \alpha x^{p^s} + \beta x + \tr(k c^{r})$.
The result then follows from \cref{thm_main,mto1_Lbi}. 
\end{proof}

 
\cref{thm_r_ps1} generalizes 
\cite[Theorems~3.1 and 3.14]{ChenK+25},  
\cite[Propositions~2]{LiCao2389}, and 
\cite[Proposition~2]{LLi18}, 
where $m = a = - b = 1$. 
It also generalizes \cite[Theorem~5.8]{2to1-YuanZW21} 
(where $m = 2$, $a = b = 1$, $u = 0$, and $v \in \mathbb{F}_{q}^{*}$) and  
\cite[Propositions~4.22 and~4.26]{nto1-NiuLQL23} 
(where $m = 3$ or $2^s$, $a = -b = v = 1$, and $u = 0$).

\section{Inverses and involutions}\label{sec_inv}

For any $1$-to-$1$ polynomial mapping $f(x)$ over $\mathbb{F}_{q}$, 
there exists a polynomial $f^{-1}(x) \in \mathbb{F}_{q}[x]$ such that 
$f^{-1}(f(c)) = c$ for each $c \in \mathbb{F}_{q}$, and $f^{-1}(x)$ 
is unique in the sense of reduction modulo $x^q - x$.
We call $f^{-1}(x)$ the \emph{inverse} of~$f(x)$ on $\mathbb{F}_{q}$. 
A polynomial $f(x) \in \mathbb{F}_{q}[x]$ is called an \emph{involution} of~$\mathbb{F}_{q}$ 
if $f(f(c)) = c$ for any $c \in \mathbb{F}_{q}$, i.e.,
$f(x)$ permutes~$\mathbb{F}_{q}$ and its inverse is itself.
Recently, finding inverses of permutation polynomials 
has attracted considerable attention; see, for example,
\cite{ZhengRDP,Zheng7deg,NiuLQW21,Yuan24Local,Wang25survey}.  
In this section, we obtain the inverses of the $1$-to-$1$ mappings obtained in this paper, 
and these results generalize several results in \cite{WuYGL25Fq2}.  
Moreover, from those $2$-to-$1$ mappings obtained in this paper, 
we also give an explicit construction of involutions associated with them.  

\subsection{Inverses of one-to-one mappings}

\begin{theorem}\label{thm_main_gen_inv}
Let $a, b, c, u, v \in\mathbb{F}_{q^2}$ satisfy $a^{q+1} = b^{q+1}$ 
and $a v \ne b u$. Let $h \in \mathbb{F}_{q^2}[x]$~and
\[
  f(x) = h(a x^q + b x + c) + u x^q + v x.
\]
Assume $a = d e^q$ and $b = d e$ 
for some~$d$, $e \in \mathbb{F}_{q^2}^{*}$  
and let $H(x) = h(d x + c) + (u / e^q) x$.
If $f(x)$ is $1$-to-$1$ on $\mathbb{F}_{q^2}$,
then its inverse on $\mathbb{F}_{q^2}$~is 
\[
  f^{-1} (x) =  a A^{-1} H 
  \big(g^{-1}(\tr(k x)) \big) - a A^{-1} x,
\]
where $A = b u - a v$, $k = e^q u^q - e v^q$,
and $g^{-1}(x)$ is the inverse of $g(x)$ in \cref{eq_g_gen} on $\mathbb{F}_{q}$.
\end{theorem}

\begin{proof}
Let $\phi(x) = (\tr(e x), x - \tr(e x))$,
$\bar{\phi}(x) = (\tr(k x), x - \tr(k x))$, and
\[
\psi (y, z)  = (g(y),\ H(y) - g(y) - a^{-1} A (y+z)).
\]
By \cref{eq_Lf=gL}, we have $\tr(k f(x)) = g(\tr(e x))$.
Note that $f(x) = H(\tr(e x)) - a^{-1} A x$.
Hence $\bar{\phi} \circ f = \psi \circ \phi$, i.e.,
the following diagram is commutative:
\begin{equation*}
\xymatrix{
    \mathbb{F}_{q^2} \ar[r]^{f}\ar[d]_{\phi}   & \mathbb{F}_{q^2} \ar[d]^{\bar{\phi}} \\
    \phi(\mathbb{F}_{q^2}) \ar[r]^{\psi}       & \bar{\phi}(\mathbb{F}_{q^2}).
}
\end{equation*}
A direct computation gives that
$\phi^{-1}(y, z) = y + z$ and
\[
  \psi^{-1} (y, z) = (g^{-1}(y),\ 
  a A^{-1} \big( H(g^{-1}(y)) - y - z \big) - g^{-1}(y)).
\]
Substituting $\bar{\phi}$, $\psi^{-1}$, and $\phi^{-1}$ into 
$f^{-1} = \phi^{-1} \circ \psi^{-1} \circ \bar{\phi}$ 
gives the desired result. 
\end{proof}

This theorem converts the computation of $f^{-1}(x)$ over $\mathbb{F}_{q^2}$
into the computation of $g^{-1}(x)$ over $\mathbb{F}_{q}$.
Assume $g(x) \in \mathbb{F}_{q}[x]$ and $\alpha$, $\beta$, $\gamma \in \mathbb{F}_{q}$ 
with $\alpha \ne 0$.
Then $g(x)$ is $1$-to-$1$ on $\mathbb{F}_{q}$ if and only if so is
$\bar{g}(x) = \alpha g(x+\beta) + \gamma$.
By choosing $\alpha$, $\beta$, $\gamma$ suitably, 
we can obtain $\bar{g}(x)$ in {\it normalized form},
that is, $\bar{g}(x)$ is monic, $\bar{g}(0) = 0$, 
and when the degree~$d$ of~$\bar{g}(x)$
is not divisible by the characteristic of $\mathbb{F}_{q}$,
the coefficient of $x^{d-1}$ is~$0$. 
It then suffices to compute the inverse of the associated normalized polynomial $\bar{g}(x)$. 
In the preceding sections, every $g_r(x)$ has degree $\le 4$ or is a linearized binomial,
and its inverse can be completely determined using the results listed in
\cite{Zheng5deg,CoulterH04,Wu-L-bi}. 
Substituting these inverses into \cref{thm_main_gen_inv}, 
we can obtain the inverses of all $1$-to-$1$ mappings $f_r(x)$ in this paper. 
For example, Theorems~3.6--3.9 of \cite{WuYGL25Fq2} 
are special cases of \cref{thm_main_gen_inv}. 

For the convenience of the reader, we include the table of the inverses 
of all normalized PPs of degree $\leq 5$ (\cref{tabel_smalldeg}), 
and an explicit formula for the inverse of a linearized permutation binomial (\cref{L-bi}).  
In particular, we state the following corollary for $h(x) = x^2$ as an explicit demonstration. 

\begin{corollary}
Let $a$, $b$, $c$, $u$, $v \in\mathbb{F}_{q^2}$ satisfy 
$a^{q+1} = b^{q+1}$ and $a v \ne b u$. Let
\begin{align*}
    \alpha = a^q B^q + b B, \quad
    \beta  = 2 B c + 2 B^q c^q + u^{q+1} - v^{q+1}, \quad
    C = A c^{2q} + B c^2. 
\end{align*} 
If $f(x) = (a x^q + b x + c)^2 + u x^q + v x$ 
is $1$-to-$1$ on $\mathbb{F}_{q^2}$,
then its inverse on $\mathbb{F}_{q^2}$ is  
\[
  f^{-1} (x) = a A^{-1} \big((\phi(x) + c)^2 
  + a^{-1} u \phi(x) - x \big),
\]
where
\[
\phi(x) =
\begin{cases}
  \beta^{-1} (A x^q + B x - C) 
        & \text{if $\alpha = 0$ and $\beta \ne 0$,}\\
  a \big(b^{-1} \alpha^{-1} (A x^q + B x - C)\big)^{q/2} 
        & \text{if $\alpha \ne 0$, $\beta = 0$, and~$q$ is even.}
\end{cases}
\]
\end{corollary}

\begin{proof}
By the proof of \cref{thm_r=2}, we have  
$g(x) = e^{-1} \alpha x^2 + \beta x + \gamma$,
where $\gamma = \tr(k c^2)$.
Write $\bar{\lambda}(x) = \tr(k x)$
and $\ell(x) = A x^q + B x - C$. 
Since $k = B / d$ and $k^q = A / d$, we get
\[
\bar{\lambda}(x) - \gamma 
= \tr(k (x - c^2))
= (B/d) (x - c^2) + (A/d) (x - c^2)^q
= \ell(x) / d.
\]
If $\alpha = 0$ and $\beta \ne 0$, 
then $g(x) = \beta x + \gamma$.
Thus $g^{-1}(x) = \beta^{-1}(x-\gamma)$ and  
$d g^{-1}(\bar{\lambda}(x)) = \beta^{-1} \ell(x)$.
If $\alpha \ne 0$, $\beta = 0$, and~$q$ is even, 
then $g(x) = e^{-1} \alpha x^2 + \gamma$.
By \cref{tabel_smalldeg}, the inverse of $x^2$ on $\mathbb{F}_{q}$ is $x^{q/2}$.
Hence $g^{-1}(x) = (e \alpha^{-1} (x-\gamma))^{q/2}$  
and by $d^{q-1} = b^q / a$ we obtain
\begin{align*}
    d g^{-1}(\bar{\lambda}(x))
    & = d (e d^{-1} \alpha^{-1} \ell(x) )^{q/2}
      = (d / d^q) (e d \alpha^{-1} \ell(x))^{q/2}\\
    & = (a / b^q) (b \alpha^{-1} \ell(x))^{q/2}
     = a (b^{-1} \alpha^{-1} \ell(x))^{q/2}.
\end{align*}
The result then follows from \cref{thm_main_gen_inv}.
\end{proof}

 \begin{table}[htbp] 
    \caption{All normalized $1$-to-$1$ polynomials of degree $\leq 5$ over $\mathbb{F}_{q}$ and their inverses \cite{Zheng5deg}}
    \label{tabel_smalldeg}
    \vspace{2pt}  
    \centering
    \renewcommand\arraystretch{1.2} 
\begin{threeparttable}
\begin{tabular}{l l l l} 
  \toprule[1pt]
  Normalized $1$-to-$1$ $g(x)$ over $\mathbb{F}_{q}$   & Inverse $g^{-1}(x)$ & $q$~~($n \geq 1$) \\
  \midrule
  $x$    & $x$        & any $q$ \\
  $x^2$  & $x^{q/2}$  & $q = 2^n$ \\
  $x^3$  & $x^{(aq-a+1)/3}$~~$(a \equiv 1-q \pmod 3)$
        & $q \not\equiv 1 \pmod 3$  \\
  $x^3 -ax$ ($a$ not a square)
    & $\displaystyle \sum_{i=0}^{n-1}a^{-\frac{3^{i+1} -1}{2}}x^{3^i}$
    & $q = 3^n$    \\
  $x^4$  &  $x^{q/4}$  & $q = 2^n$   \\
  $x^4 \pm 3x$ & $\mp(x^4 -3x)$ & $q=7$ \\
  $x^4 +ax$ ($a$ not a cube)
    & $a^{\frac{q-1}{3}}(1+a^{\frac{q-1}{3}})^{-1} \displaystyle
    \sum_{i=0}^{n-1}a^{-\frac{4^{i+1}-1}{3}}x^{4^i}$
    & $q = 2^{2n}$ \\
  $x^4 +bx^2 +ax$ ($ab \neq 0, S_{n} + aS_{n-2}^2 =1$)\tnote{*}
    & $\displaystyle \sum_{i=0}^{n-1}\big(S_{n-2-i}^{2^{i+1}}+a^{1 -2^{i+1}}S_{i}\big) x^{2^i}$
    & $q = 2^n$  \\
  $x^5$ & $x^{(aq-a+1)/5}$~~$(a \equiv (1-q)^3 \pmod{5})$
        & $q \not\equiv 1 \pmod 5$  \\
  $x^5 +ax$ ($a^2 =2$)  & $x^5 +ax$ & $q=9$\\
  $x^5 -ax$ ($a$ not a fourth power)
    & $a^{\frac{q-1}{4}}(1-a^{\frac{q-1}{4}})^{-1} \displaystyle
    \sum_{i=0}^{n-1}a^{-\frac{5^{i+1} -1}{4}}x^{5^i}$
    & $q = 5^n$   \\
  $x^5 \pm 2x^2$ & $x^5 \mp 2x^2$ & $q=7$\\[3pt]
  $x^5 +ax^3  +3a^2x$ ($a$ not a square)
    & $-a^2x^9 -ax^7 +4x^5 +4a^5x^3 -5a^4x$ & $q=13$ \\[3pt]
  $x^5 + ax^3 +5^{-1}a^2x$ ($a \neq 0$)
    & $\displaystyle \sum_{i=1}^{\lfloor k/2\rfloor}\frac{k}{k-i}\binom{k-i}{i}
    (5^{-1}a)^{5i}x^{k - 2i}$\tnote{\dag}
    & $q \equiv \pm 2 \pmod 5$ \\[12pt]
  $x^5 -2ax^3 +a^2x$ ($a$ not a square)
    & $\displaystyle \sum_{i=0}^{n-1}\sum_{j=0}^{n-1}
         2a^{\frac{q -5^{i+1} -5^{j+1} +1}{4}}
         x^{\frac{q +5^i +5^j -1}{2}}$
    & $q = 5^n$    \\[6pt]
  $x^5 +ax^3 \pm x^2 +3a^2x$ ($a$ not a square)
    & $x^5 \pm (2ax^4 - 2x^2) +a^2x^3 +ax$ & $q=7$ \\[3pt]
  \bottomrule[1pt]
  \end{tabular}
  \begin{tablenotes}
        \footnotesize
    \item[*] 
        The sequence $\{S_{i}\}$ is defined as follows: $S_{-1} =0$, $S_{0} =1$,
        $S_{i} = b^{2^{i-1}}S_{i-1} +a^{2^{i-1}}S_{i-2}$ for $1 \leq i \leq n$.
    \item[\dag] The notation $\lfloor k/2\rfloor$ denotes the largest integer $\leq k/2$ and $k = (3q^2-2)/5$.
  \end{tablenotes}
\end{threeparttable}
\end{table}

\begin{proposition}[\cite{CoulterH04, Wu-L-bi}]\label{L-bi}
Let $L(x) =x^{q^r} -ax$, where $a \in \mathbb{F}_{q^n}^{*}$ and $1 \leq r \leq n-1$.
Then $L(x)$ is $1$-to-$1$ on $\mathbb{F}_{q^n}$ if and only if $\mathrm{N}_{q^n/q^d}(a) \neq 1$,
where $d = \gcd(n, r)$. In this case, its inverse on $\mathbb{F}_{q^n}$ is
\[
L^{-1}(x)=\frac{\mathrm{N}_{q^n/q^d}(a)}{1-\mathrm{N}_{q^n/q^d}(a)}
\sum_{i=0}^{n/d -1}a^{-\frac{q^{(i+1)r} -1}{q^r-1}}x^{q^{i r}}.
\]
\end{proposition}

\subsection{Involutions from two-to-one mappings}

We next introduce a way of constructing involutions from $2$-to-$1$ mappings,
which is equivalent to that in {\cite[Theorem~3.1]{2to1-YuanZW21}}.

\begin{definition}\label{defn_ff=I}
Let $f(x) \in \mathbb{F}_{q}[x]$ be $2$-to-$1$ on $\mathbb{F}_{q}$ with $q$ even. 
For any $a$, $b \in \mathbb{F}_{q}$ with $a \neq b$ and $f(a) = f(b)$, 
there exists a polynomial $I_f(x) \in \mathbb{F}_{q}[x]$ such that 
$I_f(a) = b$ and $I_f(b) = a$, and $I_f(x)$ 
is unique in the sense of reduction modulo $x^q - x$.
We call $I_f(x)$ the involution of~$f(x)$ on $\mathbb{F}_{q}$. 
\end{definition}

The involution of a $2$-to-$1$ $f(x)$ on~$\mathbb{F}_{q}$ 
is a permutation of~$\mathbb{F}_{q}$ without fixed points.

\begin{theorem}\label{thm_invo}
Let $a, b, c, u, v \in\mathbb{F}_{q^2}$ satisfy $a^{q+1} = b^{q+1}$ 
and $a v \ne b u$. Let $h \in \mathbb{F}_{q^2}[x]$~and
\[
  f(x) = h(a x^q + b x + c) + u x^q + v x.
\]
Assume $a = d e^q$ and $b = d e$ 
for some~$d$, $e \in \mathbb{F}_{q^2}^{*}$ and let
$H(x) = h(d x + c) + (u / e^q) x$.
If $f(x)$ is $2$-to-$1$ on $\mathbb{F}_{q^2}$, then $q$ is even 
and the involution of $f(x)$ on $\mathbb{F}_{q^2}$ is given by 
\[
  I_{f} (x) = a A^{-1} f (x) 
    + a A^{-1} H \big( I_{g}( \tr(e x)  ) \big),
\]
where $A = b u - a v$ and $I_{g}(x)$ is the involution 
of $g(x)$ in \cref{eq_g_gen} on $\mathbb{F}_{q}$. 
In particular, if $h(x) = x^{r}$ and 
$I_{g}(x) = x + \alpha$, then
\begin{equation}\label{eq_Ifr}
  I_{f}(x) 
  = a A^{-1} (ax^q + bx + c)^{r} 
    + a A^{-1} (ax^q + bx + c + d \alpha)^{r} 
    + x +  A^{-1} d u \alpha.
\end{equation}
\end{theorem}

\begin{proof}
By \cref{thm_main_gen}, $f$ is $2$-to-$1$ on $\mathbb{F}_{q^2}$ 
if and only if $q$ is even and~$g$ is $2$-to-$1$ on $\mathbb{F}_{q}$. 
If~$f$ is $2$-to-$1$ on $\mathbb{F}_{q^2}$, then for any $x \in \mathbb{F}_{q^2}$ 
there is a unique $y \in \mathbb{F}_{q^2}$ with $y \neq x$ such that 
\begin{equation}\label{eq_fx=fy2}
    H(\tr(e x) ) + a^{-1} A x 
    = f(x) = f(y) =
    H(\tr(e y) ) + a^{-1} A y.
\end{equation}
Thus $\tr(e y)  \ne \tr(e x)$. By \cref{eq_Lf=gL}, we have 
$\tr(k f(x)) = g (\tr(e x))$. Then
\[
    g(\tr(e x)) = \tr(k f(x)) 
    = \tr(k f(y)) = g(\tr(e y) ).  
\]
Since~$g$ is $2$-to-$1$ on $\mathbb{F}_{q}$, we have 
$\tr(e y) = I_{g}(\tr(e x) )$ by \cref{defn_ff=I}. 
Substituting $\tr(e y) = I_{g}(\tr(e x) )$ and $y =I_{f}(x)$
into \cref{eq_fx=fy2} yields the desired result.
\end{proof}

This theorem converts the computation of $I_{f}(x)$ 
into the computation of $I_{g}(x)$.
In the previous sections, each $2$-to-$1$ polynomial 
$g_r(x)$ has the form $L(x) + \delta$, 
and thus $I_{g_r} (x) = x + \alpha$,
where $L(x)$ is a $2$-to-$1$ linearized binomial over $\mathbb{F}_{q}$ 
and $\alpha$ is the nonzero root of $L(x)$ in $\mathbb{F}_{q}$. 
Hence the involution of each $2$-to-$1$ polynomial $f_r(x)$ 
in this paper is given by~\cref{eq_Ifr}. 
In addition, the involutions in 
{\cite[Corollaries 5.11 and 5.17]{2to1-YuanZW21}} 
are also special cases of~\cref{eq_Ifr},
and the involution in 
{\cite[Theorem~11]{MesYZ231315}} 
is a special case of \cref{thm_invo}.

\section*{Acknowledgments}
The authors thank the anonymous reviewers 
for their valuable comments and suggestions, 
which helped improve the quality of the manuscript.



\bibliographystyle{elsarticle-num-names-alpha.bst} 
\bibliography{ff}

\end{document}